\let\ifLONG\iffalse
\let\ifSHORT\iftrue
\documentclass{llncs}
\usepackage[latin1]{inputenc}
\usepackage{amsmath}
\usepackage{amsfonts}
\usepackage{amssymb}
\usepackage{xcolor}
\usepackage{graphicx}
\usepackage{pst-all}
\usepackage{gastex}
\usepackage{times}
\usepackage{url}
\usepackage{algorithm}
\usepackage{algorithmic}
\def\epsilon{\varepsilon}

 \newenvironment{Theorem}{\begin{theorem}\upshape}{\end{theorem}}

 \newenvironment{Lemma}{\begin{lemma}\upshape}{\end{lemma}}

\renewenvironment{proof}{
        \noindent {\bf Proof: }}{ \qed}

\DeclareMathOperator*{\opttwo}{\mathsf{opt}} % opt with proper limits placement
\DeclareMathOperator*{\argopttwo}{\arg\mathsf{opt}} 

\DeclareMathOperator{\unif}{unif} 

\newcommand{\R}{\mathbb{R}} % real numbers
\newcommand{\Q}{\mathbb{Q}} % rational numbers

\newcommand{\error}{{\mathcal E}}
\newcommand{\locations}{{L}}
 
\newcommand{\ratematrix}{\mathbf{R}}
\newcommand{\probabilitymatrix}{\mathbf{P}}
\newcommand{\opt}{\ensuremath{{\mathsf{opt}}}}

\newcommand{\dist}{\mathit{Dist}}
\newcommand{\paths}{\mathit{Paths}}

\newcommand{\M}{\ensuremath{\mathcal{M}}}
\renewcommand{\S}{\ensuremath{\mathcal{S}}}

\begin{document}
 \sloppy

\title{Efficient Approximation of Optimal Control for Continuous-Time Markov
Games}
%\runningtitle{Efficient Approximation of Optimal Control for Continuous-Time Markov Games}
%\runningauthors{John Fearnley, Markus Rabe, Sven Schewe, and Lijun Zhang}

\author{John Fearnley\inst{1}, Markus Rabe\inst{2}, Sven Schewe\inst{1}, and Lijun Zhang\inst{3}}
 \institute{%
  ${}^1$Department of Computer Science, University of Liverpool, Liverpool, United Kingdom\\
  ${}^2$Department of Computer Science, Universit\"at des Saarlandes, Saarbr\"ucken, Germany\\
  ${}^3$DTU Informatics, Technical University of Denmark, Lyngby, Denmark} 

\maketitle

\begin{abstract}
We study the time-bounded reachability problem for continuous-time Markov
decision processes (CTMDPs) and games (CTMGs). Existing techniques for this
problem use discretisation techniques to break time into discrete intervals, and
optimal control is approximated for each interval separately. Current techniques
provide an accuracy of $O(\epsilon^2)$ on each interval, which leads to an
infeasibly large number of intervals. We propose a sequence of approximations
that achieve accuracies of $O(\epsilon^3)$, $O(\epsilon^4)$, and
$O(\epsilon^5)$, that allow us to drastically reduce the number of intervals
that are considered. For CTMDPs, the performance of the resulting algorithms is
comparable to the heuristic approach given by Buckholz and
Schulz~\cite{Buchholz11}, while also being theoretically justified. All of our
results generalise to CTMGs, where our results yield the first practically
implementable algorithms for this problem. We also provide positional strategies
for both players that achieve similar error bounds.
%The success of probabilistic model checking for discrete-time Markov decision
%processes and continuous-time Markov chains has led to rich academic and
%industrial applications. The analysis of their combination in continuous-time
%Markov decisions processes, however, is currently restricted to toy examples.
%This is due to the fact that current analysis techniques for time-bounded
%reachability require a running time linear in the reciprocal $\pi^{-1}$ of the
%required precision $\pi$.  For the high precision usually sought (for example,
%six to ten digits), this renders these techniques infeasible.  We introduce a
%sequence of increasingly precise approximations, which lead to memoryless near
%optimal schedulers. These approximations can be computed in time linear only in
%the square or cube root of $\pi^{-1}$, which is a significant advantage over
%previous techniques. Our results naturally extend to the analysis of
%continuous-time Markov games.
\end{abstract}

\iffalse
\begin{figure}[t]%
\begin{center} 
\vspace{-.7cm}
% \input{gastexFigure1.tex}
\includegraphics{gastexFigure1_help}
\end{center}
\caption{This figure shows a part of a simple CTMDP. In this example, the complete probability mass is initially concentrated in location $l$, indicated by the incoming arrow.
The goal region consists only of location $G$, indicated by the colour and the double line.
In $l$, the CTMDP offers the choice between two control actions, $a$ and $b$. 
Action $a$ has a smaller transition rate of $5$ than $b$, which has a higher transition rate of $1+8=9$.
If the control action $b$ is chosen, the next discrete transition will lead with a probability of $\frac{1}{9}$ to $\bot$, and with a probability of $\frac{8}{9}$ to $G$.
Intuitively, action $a$ leads less quickly but with higher probability (once the transition fires) to the goal location $G$.
When time is short, $a$ is the preferable action, while $b$ is preferable when much time is left.
We discuss efficient techniques for finding near optimal control policies for a given time bound for such CTMDPs and their extension to games.}
\label{fig:automaton}%
\end{figure}
\fi

\section{Introduction}
Probabilistic models are being used extensively in the formal analysis of
complex systems, including networked, distributed, and most recently, biological
systems. Over the past 15 years, probabilistic model checking for discrete-time
Markov decision processes (MDPs) and continuous-time Markov chains (CTMCs) has
been successfully applied to these rich academic and industrial applications
\cite{cadp,CosteHLS09,HenzingerMW09,esa}. However, the theory for
continuous-time Markov decision processes (CTMDPs), which mix the
non-determinism of MDPs with the continuous-time setting of CTMCs, is less well
developed.

This paper studies the \emph{time-bounded reachability} problem for CTMDPs and
their extension to continuous-time Markov games, which is a model with both
helpful and hostile non-determinism. This problem is of paramount importance for
model checking applications~\cite{Buchholz+Hahn+Hermanns+Zhang/11/MC_CTMDPs}.
The non-determinism in the system is resolved by providing a scheduler. The
time-bounded reachability problem is to determine or to approximate, for a given
set of goal locations~$G$ and time bound~$T$, the maximal (or minimal)
probability of reaching~$G$ before the deadline~$T$ that can be achieved by a
scheduler.

Early work on this problem focused on restricted classes of schedulers, such
schedulers without any access to time in systems with uniform transition
rates~\cite{Baier+all/05/efficientCTMDP}. Recently however, results have been
proved for the more general class of \emph{late
schedulers}~\cite{Zhang+Neuh/10/reachabilityCTMDPs}, which will be studied in
this paper. The different classes of schedulers are contrasted by
Neuh\"{a}u{\ss}er et.\ al.\ \cite{NeuhausserDelayedNondeterminism09}, and they
show that late schedulers are the most powerful class.  Several algorithms have
been given to approximate the time-bounded reachability probabilities for CTMDPs
using this scheduler
class~\cite{Buchholz+Hahn+Hermanns+Zhang/11/MC_CTMDPs,ChenHKM10,Zhang+Neuh/10/reachabilityCTMDPs,Zhang+Neuh/10/ModelCheckingIMCs}.

The current state-of-the-art techniques for solving this problem are based on
different forms of \emph{discretisation}. This technique splits the time bound
$T$ into small intervals of length $\epsilon$. Optimal control is 
approximated for each interval separately, and these approximations are combined
to produce the final result. Current techniques can approximate optimal control
on an interval of length $\epsilon$ with an accuracy of $O(\epsilon^2)$.
However, to achieve a precision of $\pi$ with these techniques, one must choose
$\epsilon \approx \pi/T$, which leads to $O(T^2/\pi)$ many intervals. Since the
desired precision is often high (it is common to require that $\pi\leq10^{-6}$),
this leads to an infeasibly  large number of intervals that must be considered
by the algorithms.

A recent paper of Buckholz and Schulz~\cite{Buchholz11} has addressed this
problem for practical applications, by allowing the interval sizes to vary. In
addition to computing an approximation of the maximal time-bounded reachability
probability, which provides a lower bound on the optimum, they also compute an
upper bound. As long as the upper and lower bounds do not diverge too far, the
interval can be extended indefinitely. In practical applications, where the
optimal choice of action changes infrequently, this idea allows their algorithm
to consider far fewer intervals while still maintaining high precision. However,
from a theoretical perspective, their algorithm  is not particularly satisfying.
Their method for extending interval lengths depends on a heuristic, and in the
worst case their algorithm may consider  $O(T^2/\pi)$ intervals, which is
not better than other discretisation based techniques.

\paragraph{\bf Our contribution.} In this paper we present a method of obtaining larger interval sizes that
satisfies both theoretical and practical concerns. Our approach is to provide
more precise approximations for each $\epsilon$ length interval. While current
techniques provide an accuracy of $O(\epsilon^2)$, we propose a sequence of
approximations, called double $\epsilon$-nets, triple $\epsilon$-nets, and
quadruple $\epsilon$-nets, with accuracies $O(\epsilon^3)$, $O(\epsilon^4)$, and
$O(\epsilon^5)$, respectively. Since these approximations are much more precise
on each interval, they allow us to consider far fewer intervals while still
maintaining high precision. For example, Table~\ref{tbl:intervals} gives the
number of intervals considered by our algorithms, in the worst case, for a
normed CTMDP with time bound $T = 10$.

\begin{table}
\label{tbl:intervals}
\begin{center}
\begin{tabular}{@{$\;$}c@{$\;$}|@{$\;\;$}c@{$\;\;$}||@{$\;\;$}c@{$\;\;$}|@{$\;\;$}c@{$\;\;$}|@{$\;\;$}c@{$\;\;$}}
Technique & Error & $\pi = 10^{-7}$ & $\pi = 10^{-9}$ & $\pi = 10^{-11}$ \\
\hline
\hline

Current techniques & $O(\epsilon^2)$ &
$1,000,000,000$ & $100,000,000,000$ & $10,000,000,000,000$ \\
\hline

Double $\epsilon$-nets & $O(\epsilon^3)$ &
$81,650$ & $816,497$ & $8,164,966$ \\
\hline

Triple $\epsilon$-nets & $O(\epsilon^4)$ &
$3,219$ & $14,939$ & $69,337$ \\
\hline

Quadruple $\epsilon$-nets & $O(\epsilon^5)$ &
$605$ & $1,911$ & $6,043$ \\

\end{tabular}
\end{center}
\caption{The number of intervals needed by our algorithms for precisions
$10^{-7}, 10^{-9}$, and~$10^{-11}$.}
\end{table}

Of course, in order to become more precise,  we must spend additional
computational effort. However, the cost of using double $\epsilon$-nets
instead of using current techniques requires only an extra factor of $\log
|\Sigma|$, where $\Sigma$ is the set of actions. Thus, in almost all cases, the
large reduction in the number of intervals far outweighs the extra cost of using
double $\epsilon$-nets. Our worst case running times for triple and quadruple
$\epsilon$-nets are not so attractive: triple $\epsilon$-nets require an extra
$|L| \cdot |\Sigma^2|$ factor over double $\epsilon$-nets, where $L$ is the
set of locations, and quadruple $\epsilon$-nets require yet another $|L| \cdot
|\Sigma^2|$ factor over triple $\epsilon$-nets. However, these worst case
running times only occur when the choice of optimal action changes frequently,
and we speculate that the cost of using these algorithms in practice is much
lower than our theoretical worst case bounds. Our experimental results with
triple $\epsilon$-nets support this claim.

An added advantage of our techniques is that they can be applied to
continuous-time Markov games as well as to CTMDPs. Buckholz and Schulz restrict
their analysis to CTMDPs. Therefore, to the best of our knowledge, we present
the first practically implementable approximation algorithms for the
time-bounded reachability problem in CTMGs. Each approximation also provides
positional strategies for both players that achieve similar error bounds.

\section{Preliminaries}\label{sect:prelim}
\begin{definition}
  A continuous-time Markov game (or simply Markov game) is 
  a tuple
$(\locations,\locations_r,\locations_s,\Sigma,\ratematrix,\probabilitymatrix,\nu)$,
consisting of 
a finite set $\locations$ of locations, which is
partitioned into 
locations $\locations_r$ (controlled by a
\emph{reachability} player) and $\locations_s$ (controlled by a
\emph{safety} player),
a finite set $\Sigma$ of actions,
a rate matrix
$\ratematrix: (\locations\times\Sigma\times\locations) \to
\Q_{\geqslant0}$,
a discrete transition matrix $\probabilitymatrix:
(\locations\times\Sigma\times\locations) \to \Q\cap
[0,1]$, and
an initial distribution $\nu\in\dist(\locations)$.
\end{definition}
We require that the following side-conditions hold:
For all locations $l\in\locations$, there must be an action $a\in\Sigma$ such that $\ratematrix(l,a,\locations):=\sum_{l'\in \locations} \ratematrix(l,a,l')>0$, which we call \emph{enabled}. 
We denote the set of enabled actions in $l$ by $\Sigma(l)$. 
For  a location $l$ and actions $a\in \Sigma(l)$, we require for all locations $l'$ that  $\probabilitymatrix(l,a,l') = \frac{\ratematrix(l,a,l')}{\ratematrix(l,a,\locations)}$, and we require $\probabilitymatrix(l,a,l')=0$ for non-enabled actions. 
We define the \emph{size} $|\M|$ of a Markov game as the number of non-zero rates in the rate matrix $\ratematrix$.

A Markov game is called \emph{uniform} with uniform transition rate $\lambda$,
if $\ratematrix(l,a,\locations)=\lambda$ holds for all locations $l$ and enabled
actions~$a\in \Sigma(l)$.  We further call a Markov game \emph{normed}, if its
uniformisation rate is $1$. Note that for normed Markov games we have
$\ratematrix=\probabilitymatrix$. We will present our results for normed Markov
games only. The following lemma states that our algorithms for normed Markov
games can be applied to solve general Markov games.

\begin{lemma}
\label{lem:uniform}
We can adapt an $O(f(\M))$ time algorithm for normed Markov games to solve
general Markov games in time $O(f(\M) + |L|)$.
\end{lemma}

We are particularly interested in Markov games with a single player, which are continuous-time Markov decision processes (CTMDPs).
In CTMDPs all positions belong to the reachability player
($\locations=\locations_r$), or to the safety player
($\locations=\locations_s$), depending on whether we analyse the
\emph{maximum} or \emph{minimum} reachability probability problem.

\begin{figure}[t]
\begin{center}
\begin{picture}(55,25)(-1,-13)
 \gasset{ELdist=0.4}
 \node[Nmarks=i,iangle=180](s0)(0,0){$l_S$}
  \node(bot)(0,20){$\bot$}
  \node[Nh=6,Nmr=1](s1)(25,20){$l_R$}
  \node[Nmarks=r](s2)(50,0){$G$} 
  \node(s3)(50,20){$l$}
 \drawedge(s0,s2){$b,\frac{1}{8}$}
 \drawedge(s0,bot){$b,\frac{7}{8}$}
  \drawedge(s0,s1){$a,1$}
 \drawedge(s3,s2){$\frac{1}{10}$}
  \drawedge[ELside=l](s1,s2){$a,\frac{1}{20}$} 
  \drawedge[ELside=r](s1,bot){$a,\frac{3}{20}$} 
  \drawedge(s1,s3){$b,\frac{1}{5}$}
\end{picture} 
\hspace{.8cm}
\scalebox{0.2}{\includegraphics{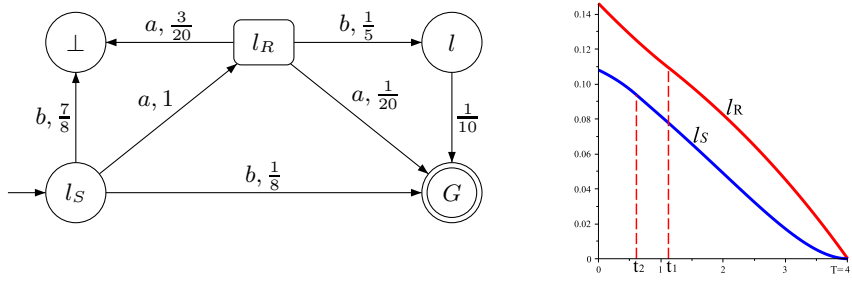}}
\end{center}
\vspace{-.8cm}
\caption{Left: a normed Markov game. Right: the function $f$ within $[0,4]$ for $l_R$ and $l_S$.}
\label{fig:game_example}\vspace{-3.75mm}
\end{figure}

As a running example, we will use the normed Markov game shown in the left half
of Figure~\ref{fig:game_example}. Locations belonging to the safety player are
drawn as circles, and locations belonging to the reachability player are drawn
as rectangles. The self-loops of the normed Markov game are omitted. The
locations $G$ and $\bot$ are absorbing, and there is only a single enabled
action for $l$.  It therefore does not matter which player owns $l$, $G$, and
$\bot$. 

\vspace{-1mm}\subsubsection{Schedulers and Strategies}
We consider Markov games in a time interval $[0,T]$ with $T\in\R_{\geq0}$.
The non-determinism in the system needs to be resolved by a pair of
strategies for the two players which together form a \emph{scheduler}
for the whole system.  
Formally, a strategy is a function in $\paths_{r/s}\times[0,T]\rightarrow
\Sigma$, where $\paths_{r}$ and $\paths_s$ are the sets of finite paths
$l_0\xrightarrow{a_0,t_0}l_1\dots\xrightarrow{a_{n-1},t_{n-1}}l_n$ with $l_n\in
L_r$ and $l_n\in L_s$, respectively, and we use $\S_r$ and $\S_s$ to denote the
strategies of reachability player and the strategies of safety player,
respectively. 
(For technical reasons one has to restrict the schedulers to those which are
measurable. This restriction, however, is of no practical relevance. In
particular, simple piecewise constant timed-positional strategies
$\locations\times[0,T]\to\Sigma$ suffice for optimal
scheduling~\cite{Rabe+Schewe/11/optimalSchedulersExist,Zhang+Neuh/10/reachabilityCTMDPs,Bellman/57/DP},
and all schedulers that occur in this paper are from the particularly tame class
of cylindrical schedulers~\cite{Rabe+Schewe/11/optimalSchedulersExist}.) 

If we fix a pair $(\S_r,\S_s)$ of strategies,  we obtain a deterministic
stochastic process, which is in fact a time inhomogeneous Markov chain, and we
denote it by $\M_{\S_{r,s}}$. For $t\le T$, we use $Pr_{\S_{r+s}}(t)$ to denote
the transient distribution at time $t$ over $S$ under the scheduler $(\S_r,
\S_s)$.

Given a Markov game $\M$, a goal region $G\subseteq \locations$, and
 a time bound $T\in\R_{\geq0}$, we are interested in the
 \emph{optimal} probability of being in a goal state at time $T$ (and the
corresponding pair of optimal strategies). This is given by:\vspace{-.5mm}
 $$\sup_{\S_r\in \text{TP}} \inf_{\S_s\in \text{TP}}
 \sum_{l\in G}Pr_{\S_{r+s}}(l,T),\vspace{-1mm}$$  where
 $Pr_{\S_{r+s}}(l,T):=Pr_{\S_{r+s}}(T)(l)$.  It is commonly referred
 to as the \emph{maximum} time-bounded reachability probability
 problem in the case of CTMDPs with a reachability player only.  For
 $t\le T$, we define $f:\locations\times\R_{\geq0}\to[0,1]$, to be the
 optimal probability to be in the goal region at the time bound $T$,
 assuming that we start in location $l$ and that $t$ time units have
 passed already.  By definition, it holds then that $f(l,T)=1$ if
 $l\in G$ and $f(l,T)=0$ if $l\not\in G$.  Optimising the vector of
 values $f(\cdot,0)$ then yields the optimal value and its \emph{optimal piecewise deterministic strategy}.

Let us return to the example shown in Figure~\ref{fig:game_example}. The right
half of the Figure shows the optimal reachability probabilities, as given by
$f$, for the locations~$l_R$ and~$l_S$ when the time bound $T = 4$. The
points~$t_1 \approx 1.123$ and~$t_2 \approx 0.609$ represent the times at which
the optimal strategies change their decisions. Before~$t_1$ it is optimal for
the reachability player to use action~$b$ at~$l_R$, but afterwards the optimal
choice is action~$a$. Similarly, the safety player uses action~$b$ before~$t_2$,
and switches to~$a$ afterwards.

\vspace{-1mm}\subsubsection{Characterisation of $f$}
We define a matrix $\mathbf{Q}$ such that
$\mathbf{Q}(l,a,l') = \mathbf{R}(l,a,l')$ if $l'\neq l$ and
$\mathbf{Q}(l,a,l) = - \sum_{l'\neq l}\mathbf{R}(l,a,l')$.
The optimal function~$f$ can be characterised as the following set of
differential  equations  \cite{Bellman/57/DP}, see also
\cite{Miller/68/finiteStrategies,Martin67}. For each $l\in L$ we define
$f(l,T) = 1$ if $l\in G$, and $0$ if $l\not\in G$. Otherwise, for $t < T$, we
define:\vspace{-1mm}
\begin{equation}
  \label{eq:diffeq}
-\dot{f}(l,t) = 
\opttwo_{a\in\Sigma(l)}
\sum_{l'\in\locations} \mathbf{Q}(l,a,l') \cdot f(l',t),\vspace{-1mm}
\end{equation}
where $\opt \in \{\max,\min\}$ is $\max$ for reachability player
locations and $\min$ for safety player locations.  We will use the
$\opt$-notation throughout this paper.

Using the matrix $\ratematrix$, Equation \eqref{eq:diffeq} can be rewritten to:
\begin{equation}
  \label{eq:diffeqAlt}
-\dot{f}(l,t) = 
\opttwo_{a\in\Sigma(l)}
\sum_{l'\in\locations} \ratematrix(l,a,l') \cdot \left(f(l',t) -f(l,t)\right)
\end{equation}

For uniform Markov games, we simply have $\mathbf{Q}(l,a,l) =
\mathbf{R}(l,a,l) - \lambda$, with $\lambda = 1$ for normed Markov
games.  This also provides an intuition for the fact that
uniformisation does not alter the reachability probability: the rate
$\ratematrix(l,a,l)$ does not appear in~\eqref{eq:diffeq}.

\section{Approximating Optimal Control for Normed Markov Games}\label{sect:fishing}

In this section we describe $\epsilon$-nets, which are a technique for
approximating optimal values and strategies in a normed continuous-time Markov
game. Thus, throughout the whole section, we fix a normed Markov game
$\M=(\locations, \locations_r, \locations_s, \Sigma, \ratematrix,
\probabilitymatrix, \nu)$. 

Our approach to approximating optimal control within the Markov game is to break
time into intervals of length~$\epsilon$, and to approximate optimal control
separately in each of the $\lceil\frac{T}{\varepsilon}\rceil$ distinct
intervals. Optimal time-bounded reachability probabilities are then computed
iteratively for each interval, starting with the final interval and working
forwards in time. The error made by the approximation in each interval is called
the \emph{step error}. In Section~\ref{sec:steperror} we show that if the step
error in each interval is bounded, then the \emph{global error} made by our
approximations is also bounded.

Our results begin with a simple approximation that finds the optimal action at
the start of each interval, and assumes that this action is optimal for the
duration of the interval. We refer to this as the \emph{single} $\epsilon$-net
technique, and we will discuss this approximation in
Section~\ref{sec:simple_nets}. While it only gives a simple linear function as
an approximation, this technique gives error bounds of $O(\epsilon^2)$, which is
comparable to existing techniques.  

However, single $\epsilon$-nets are only a starting point for our results. Our
main observation is that, if we have a piecewise polynomial approximation of
degree $c$ that achieves an error bound of $O(\epsilon^k)$, then we can compute
a piecewise polynomial approximation of degree $c+1$ that achieves an error
bound of $O(\epsilon^{k+1})$. Thus, starting with single $\epsilon$-nets, we can
construct double $\epsilon$-nets, triple $\epsilon$-nets, and quadruple
$\epsilon$-nets, with each of these approximations becoming increasingly more
precise. The construction of these approximations will be discussed in
Sections~\ref{sec:double_nets} and~\ref{ssec:beyond}.

In addition to providing an approximation of the time-bounded reachability
probabilities, our techniques also provide positional strategies for both
players. For each level of $\epsilon$-net, we will define two approximations:
the function $p_1$ is the approximation for the time-bounded reachability
probability given by single $\epsilon$-nets, and the function $g_1$ gives the
reachability probability obtained by following the positional strategy that is
derived from $p_1$. This notation generalises to deeper levels of
$\epsilon$-nets: the functions $p_2$ and $g_2$ are produced by double
$\epsilon$-nets, and so on.

We will use $\error(k,\epsilon)$ to denote the difference between~$p_k$ and~$f$.
In other words, $\error(k,\epsilon)$ gives the difference between the
approximation~$p_k$ and the true optimal reachability probabilities. We will use
$\error_s(k,\epsilon)$ to denote the difference between~$g_k$ and~$f$. We defer
formal definition of these measures to subsequent sections. Our objective in the
following subsections is to show that the step errors $\error(k,\varepsilon)$
and $\error_s(k,\varepsilon)$ are in $O(\varepsilon^{k+1})$, with small
constants.

\subsection{Step Error and Global Error}
\label{sec:steperror}

In subsequent sections we will prove bounds on the \emph{$\epsilon$-step} error
that is made by our approximations. This is the error that is made by our
approximations in a single interval of length~$\epsilon$. However, in order for
our approximations to be valid, they must provide a bound on the \emph{global}
error, which is the error made by our approximations over every $\epsilon$
interval. In this section, we prove that, if the $\epsilon$-step error of an
approximation is bounded, then the global error of the approximation is bounded
by the sum of these errors.

We define $f:[0,T]\rightarrow [0,1]^{|L|}$ as the vector valued
function $f(t) \mapsto \bigotimes_{l\in\locations}f(l,t)$ that maps each point
of time to a vector of reachability probabilities, with one entry for each
location. Given two such vectors $f(t)$ and $p(t)$, we define the maximum norm $\|f(t)
- p(t)\|= \max\{|f(l,t)-p(l,t)| \mid l \in \locations\}$, which gives the
largest difference between~$f(l, t)$ and~$p(l, t)$.

We also introduce notation that will allow us to define the values at the start
of an $\epsilon$ interval. For each interval $[t-\varepsilon,t]$, we define
$f_{x}^t:[t-\varepsilon,t]\rightarrow [0,1]^{|L|}$ to be the function obtained
from the differential equations~(\ref{eq:diffeq}) when the values at the
time~$t$ are given by the vector $x\in [0,1]^{|L|}$. More formally, if $\tau =
t$ then we define $f_{x}^t(\tau) = x$, and if $t-\varepsilon \le \tau < t$ and
$l\in L$ then we define:
\begin{equation}
  \label{eq:2}
-\dot{f}_{x}^t(l,\tau)=\opttwo_{a\in\Sigma(l)}\sum_{l'\in L} \mathbf{Q}(l,a,l')
f_x^t(l',\tau).
\end{equation}

The following lemma states that if the $\epsilon$-step error is
bounded for every interval, then the global error is simply the sum of these
errors.  

\begin{lemma}
\label{lem:individual}
Let $p$ be an approximation of~$f$ that satisfies $\|f(t) - p(t)\| \leq \mu$ for some
time point $t \in [0, T]$. If $\|f_{p(t)}^t(t-\varepsilon) -
p(t-\varepsilon)\|\leq \nu$ then we have $\|f(t-\varepsilon) -
p(t-\varepsilon)\|\leq \mu+\nu$.
\end{lemma}

\subsection{Single $\varepsilon$-Nets}
\label{sec:simple_nets}

In single $\epsilon$-nets, we compute the gradient of the function~$f$ at the
end of each interval, and we assume that this gradient remains constant
throughout the interval. This yields a \emph{linear} approximation
function~$p_1$, which achieves a local error of $\varepsilon^2$. 

We now define the function $p_1$. For initialisation, we define $p_1(l,T)=1$ if
$l\in G$ and $p_1(l,T)=0$ otherwise. Then, if~$p_1$ is defined for the interval
$[t, T]$, we will use the following procedure to extend it to the interval $[t -
\epsilon, T]$. We first determine the optimising enabled actions for each
location for $f_{p_1(t)}^t$ at time $t$. That is, we choose, for all $l\in
\locations$ and all $a \in \Sigma(l)$, an action:

\begin{equation}\label{eq:opt_act}
a_l^t \in 
\argopttwo_{a\in\Sigma(l)}
\sum_{l'\in L}\mathbf{Q}(l,a,l') \cdot p_1(l',t).
\end{equation}
We then fix $c_l^t = \sum_{l'\in L}\mathbf{Q}(l,a_l^t,l') \cdot p_1(l',t)$ as
the descent of $p_1(l,\cdot)$ in the interval $[t-\varepsilon,t]$. Therefore,
for every $\tau \in [0, \epsilon]$ and every $l \in \locations$ we have:

\begin{equation*}
-\dot{p}_1(l,t-\tau) = c_l^t \quad \mbox{ and } \quad
p_1(l,t-\tau)=p_1(l,t)+\tau \cdot c_l^t.
\end{equation*}

Let us return to our running example. We will apply the approximation~$p_1$ to
the example shown in Figure~\ref{fig:game_example}. We will set
$\varepsilon=0.1$, and focus on the interval $[1.1,1.2]$ with initial values
$p_1(G,1.2)=1$, $p_1(l,1.2)=0.244$, $p_1(l_R,1.2)=0.107$, $p_1(l_S,1.2)=0.075$,
$p_1(\bot,1.2)=0$. These are close to the true values at time~$1.2$. Note that
the point~$t_1$, which is the time at which the reachability player switches the
action played at~$l_R$, is contained in the interval $[1.1, 1.2]$. Applying
Equation~\eqref{eq:opt_act} with these values allows us to show that the
maximising action at~$l_R$ is~$a$, and the minimising action at $l_S$ is
also~$a$. As a result, we obtain the approximation $p_1(l_R,t-\tau)= 0.0286\tau
+ 0.107$ and $p_1(l_S,t-\tau)= 0.032\tau +  0.075$.

We now prove error bounds for the approximation~$p_1$. Recall that $\error(1,
\tau)$ denotes the difference between $f$ and $p_1$ after $\tau$ time units. We
can now formally define this error, and prove the following bounds.

\begin{lemma}
\label{extheo:single}
If $\varepsilon \leq 1$, then $\error(1, \epsilon) :=
\|f_{p_1(t)}^t(t-\epsilon)-p_1(t-\epsilon)\| \le \epsilon^2$. 
\end{lemma}

The approximation~$p_1$ can also be used to construct strategies for the two
players with similar error bounds. We will describe the construction for
the reachability player. The construction for the safety player can be derived
analogously.

The strategy for the reachability player is to play the action chosen by~$p_1$
during the entire interval $[t - \epsilon, t]$. We will define a system of
differential equations $g_1(l, \tau)$ that describe the outcome when the
reachability fixes this strategy, and when the safety player plays an
optimal counter strategy. For each location~$l$, we define $g_1(l, t) =
f_{p_1(t)}^t(l, t)$, and we define $g_1(l, \tau)$, for each $\tau \in [t - \epsilon,
t]$, as:
\begin{align}
\label{eq:g}
-\dot{g_1}(l,\tau) &= \sum_{l'\in\locations} \mathbf{Q}(l,a_l^t,l') \cdot
g_1(l',\tau) &\text{if $l \in \locations_r$,} \\
-\dot{g_1}(l,\tau) &= 
\min_{a\in\Sigma(l)} \sum_{l'\in\locations} \mathbf{Q}(l,a,l') \cdot
g_1(l',\tau) &\text{if $l \in \locations_s$.}
\end{align}

We can prove the following bounds for $\error_s(1,\varepsilon)$, which is the
difference between $g_1$ and $f^t_{p_1(t)}$ on an interval of length $\epsilon$.

\begin{lemma}
\label{lem:single}
We have $\error_s(1,\varepsilon) :=
\|g_1(t-\varepsilon)-f^t_{p_1(t)}(t-\varepsilon)\| \le 2 \cdot \varepsilon^2$.
\end{lemma}

Lemma~\ref{extheo:single} gives the $\epsilon$-step error for $p_1$, and we can
apply Lemma~\ref{lem:individual} to show that the global error is bounded by
$\epsilon^2\cdot\frac{T}{\epsilon}=\varepsilon T$. If~$\pi$ is the required
precision, then we can choose $\epsilon=\frac{\pi}{T}$ to produce an
algorithm that terminates after $\frac{T}{\varepsilon}\approx \frac{T^2}{\pi}$
many steps. Hence, we obtain the following known result. 

\begin{theorem}
\label{thm:singlenets}
For a normed Markov game $\M$ of size $|\M|$, we can compute a $\pi$-optimal strategy and determine the quality of $\M$ up to precision $\pi$ in time $O(|\M| \cdot T \cdot \frac{T}{\pi})$.
\end{theorem}

\subsection{Double $\varepsilon$-Nets}\label{sec:double_nets}

In this section we show that only a small amount of additional computation
effort needs to be expended in order to dramatically improve over the precision
obtained by single $\epsilon$-nets. This will allow us to use much larger values
of~$\epsilon$ while still retaining our desired precision.

In single $\epsilon$-nets, we computed the gradient of~$f$ at the start of each
interval and assumed that the gradient remained constant for the duration of
that interval. This gave us the approximation~$p_1$. The key idea behind double
$\epsilon$-nets is that we can use the approximation~$p_1$ to approximate the
gradient of $f$ throughout the interval.

We define the approximation~$p_2$ as follows: we have $p_2(l, T) = 1$ if $l \in
G$ and $0$ otherwise, and if $p_2(l, \tau)$ is defined for every $l \in L$ and
every $\tau \in [t, T]$, then we define $p_2(l, \tau)$ for every $\tau \in [t -
\epsilon, t]$ as:
\begin{equation}\label{eq:approx2}
-\dot{p}_2(l,\tau) = \opttwo_{a\in\Sigma(l)}
\sum_{l'\in\locations} \ratematrix(l,a,l') \cdot \left(p_1(l',\tau) -
p_1(l,\tau)\right)\quad \forall l\in\locations.
\end{equation}
By comparing Equations~\eqref{eq:approx2} and~\eqref{eq:diffeqAlt}, we can see
that double $\epsilon$-nets uses $p_1$ as an approximation for $f$ during the
interval $[t-\epsilon, t]$. Furthermore, in contrast to~$p_1$, note that the
approximation~$p_2$ can change it's choice of optimal action during the
interval. The ability to change the choice of action during an interval is the
key property that allows us to prove stronger error bounds than previous work.%It is this property that allows us to prove better error bounds for~$p_2$.

\begin{Lemma} 
\label{lem:double}
If $\varepsilon \leq 1$ then $\error(2,\varepsilon) := \|p_2(\tau)-f_{p_2(t)}^t(\tau)\| \leq \frac{2}{3}\varepsilon^3$. 
\end{Lemma}  

Let us apply the approximation $p_2$ to the example shown in 
Figure~\ref{fig:game_example}. We will again use the interval $[1.1, 1.2]$,
and we will use initial values that were used when we applied single
$\epsilon$-nets to the example in Section~\ref{sec:simple_nets}. We will focus
on the location~$l_R$. From the previous section, we know that $p_1(l_R,t-\tau)=
0.0286\tau +  0.107$, and for the actions~$a$ and~$b$ we have:
\begin{itemize}
\item $\sum_{l'\in\locations} \ratematrix(l_R,a,l') p_1(l',t-\tau)=
\frac{1}{20} +\frac{4}{5} p_1(l_R,t-\tau),$
\item $\sum_{l'\in\locations} \ratematrix(l_R,b,l') p_1(l',t-\tau)=
\frac{1}{5}  p_1(l,t-\tau) +\frac{4}{5} p_1(l_R,t-\tau).$
\end{itemize}
These functions are shown in Figure~\ref{fig:simple_net}.
To obtain the approximation $p_2$, we must take the maximum of these two functions.
Since~$p_1$ is a linear function, we know that these two functions have exactly
one crossing point, and it can be determined that this point occurs when
$p_1(l,t-\tau)=0.25$, which happens at $\tau = z := \frac{5}{63}$. Since $z
\le 0.1 = \epsilon$, we know that the lines intersect within the interval $[1.1,
1.2]$. Consequently, we get the following piecewise quadratic function for
$p_2$:
\begin{itemize}
\item When $0\le \tau \leq z$, we use the action $a$ and obtain
$-\dot{p}_2(l_R,t-\tau) = -0.00572\tau + 0.0286$, which implies that
$p_2(l_R,t-\tau)=  - 0.00286\tau^2 + 0.0286\tau +  0.107$. 
\item When $z<\tau\le 0.1$ we use action $b$ and obtain 
$-\dot{p}_2(l_R,t-\tau) = 0.0094\tau + 0.0274$, which implies that 
$p_2(l_R,t-\tau) = 0.0047\tau^2 + 0.0274\tau + 0.107047619$.
\end{itemize}

\begin{figure}[t]
\begin{center}
\vspace{-.5cm}
    \scalebox{0.95}{\input{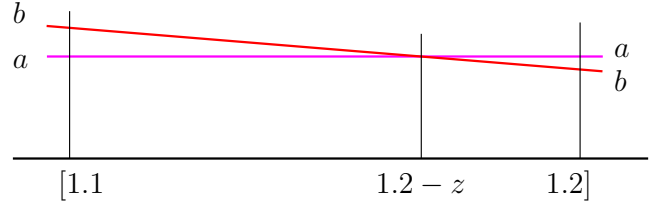}}
\vspace{-.5cm}
\end{center}
\caption{This figure shows how $-\dot{p}_2$ is computed on the interval $[1.1,
1.2]$ for the location $l_R$. The function is given by the upper envelope of the two functions: it
agrees with the quality of $a$ on the interval $[1.2-z, 1.2]$ and with the
quality of $b$ on the interval $[1.1, 1.2-z]$.}
 \label{fig:simple_net}
\end{figure}

As with single $\epsilon$-nets, we can provide a strategy that obtains similar
error bounds. Once again, we will consider only the reachability player, because
the proof can easily be generalised for the safety player. In much the same way
as we did for $g_1$, we will define a system of differential equations $g_2(l,
\tau)$ that describe the outcome when the reachability player plays according
to~$p_2$, and the safety player plays an optimal counter strategy. For each
location~$l$, we define $g_2(l, t) = f_{p_2(t)}^t(l, t)$. If $a_l^\tau$ denotes
the action that maximises Equation~\eqref{eq:approx2} at the time point $\tau
\in [t - \epsilon, t]$, then we define $g_2(l, \tau)$, as:
\begin{align}
\label{eq:g2}
-\dot{g_2}(l,\tau) &= \sum_{l'\in\locations} \mathbf{Q}(l,a_l^{\tau},l') \cdot
g_2(l',\tau) &\text{if $l \in \locations_r$,} \\
-\dot{g_2}(l,\tau) &= 
\min_{a\in\Sigma(l)} \sum_{l'\in\locations} \mathbf{Q}(l,a,l') \cdot
g_2(l',\tau) &\text{if $l \in \locations_s$.}
\end{align}
The following lemma proves that difference between~$g_2$
and~$f_{p_2(t)}^t$ has similar bounds to those shown in Lemma~\ref{lem:double}

\begin{lemma}
\label{lem:doubleConcrete}
If $\epsilon \le 1$ then we have $\error_s(2,\varepsilon) :=
\|g_2(t-\varepsilon)-f^t_{p_2(t)}(t-\varepsilon)\| 
\leq 2 \cdot \varepsilon^3$.
\end{lemma}

%We can now use the triangle inequality to obtain bounds for
%$\error_p(2,\varepsilon)$.
%\begin{Corollary}
%\label{cor:double}
%If $\varepsilon \leq 1$, then we have $\error_p(2,\varepsilon)
%:= \|g_2(t-\varepsilon)-p_2(t-\varepsilon)\|
 %\leq \frac{2}{3}\varepsilon^3$.
%\end{Corollary}

Computing the approximation~$p_2$ for an interval $[t - \epsilon, t]$ is not
expensive. The fact that~$p_1$ is linear implies that
each action can be used for at most one subinterval of $[t - \epsilon, t]$.
Therefore, there are less than~$|\Sigma|$ points at which the strategy changes,
which implies that~$p_2$ is a piecewise quadratic function with at most
$|\Sigma|$ pieces. It is possible to design an algorithm that uses sorting to
compute these switching points, achieving the following complexity.

\begin{Lemma}
\label{lem:sortdouble}
Computing $p_2$ for an interval $[t - \epsilon, t]$ takes $O(|\M| +
|\locations|\cdot|\Sigma|\cdot \log |\Sigma|)$ time.
\end{Lemma}

Since the $\epsilon$-step error for double $\epsilon$-nets is bounded by
$\varepsilon^3$, we can apply Lemma~\ref{lem:individual} to conclude that the
global error is bounded by
$\varepsilon^3\cdot\frac{T}{\varepsilon}=\varepsilon^2 T$. Therefore, if we want
to compute $f$ with a precision of~$\pi$, we should choose $\varepsilon \approx
\sqrt{\frac{\pi}{T}}$,  which gives $\frac{T}{\varepsilon}\approx
\frac{T^{1.5}}{\sqrt{\pi}}$ distinct intervals.

\begin{Theorem}
\label{thm:doublenets}
For a normed Markov game $\mathcal M$ we can approximate the time-bounded reachability,
construct $\pi$ optimal memoryless strategies for both players, and
determine the quality of these strategies with precision $\pi$ in time
$O(|\M| \cdot T
 \cdot \sqrt{\frac{T}{\pi}} + |\locations| \cdot  T
 \cdot \sqrt{\frac{T}{\pi}} \cdot |\Sigma| \log |\Sigma|)$.
\end{Theorem}

\subsection{Triple $\varepsilon$-Nets and Beyond}\label{ssec:beyond}

The techniques used to construct the approximation~$p_2$ from the
approximation~$p_1$ can be generalised. This is because the only property
of~$p_1$ that is used in the proof of Lemma~\ref{lem:double} is the fact that it
is a piecewise polynomial function that approximates~$f$. Therefore, we can
inductively define a sequence of approximations~$p_k$ as follows:
\begin{align}\label{eq:approxk}
-\dot{p}_k(l,\tau) = \opttwo_{a\in\Sigma(l)}
\sum_{l'\in\locations} \ratematrix(l,a,l') \cdot \left(p_{k-1}(l',\tau) - p_{k-1}(l,\tau)\right)
\end{align}
We can repeat the arguments from the previous sections to obtain the following
error bounds: 

\begin{Lemma}
\label{lem:knets}
For every $k > 2$, if we have $\error(k, \varepsilon) \leq c \cdot
\epsilon^{k+1}$, then we have $\error(k + 1, \varepsilon) \leq \frac{2}{k+2}
\cdot c \cdot \epsilon^{k+2}$. Moreover, if we additionally have that
$\error_s(k, \varepsilon) \leq d \cdot \epsilon^{k+1}$, then we also have that 
$\error_s(k + 1, \varepsilon) \leq \frac{8c + 3d}{k+2} \cdot \epsilon^{k+2}$.
\end{Lemma}

Computing the accuracies explicitly for the first four levels of $\epsilon$-nets
gives:

\begin{center}
\begin{tabular}{@{$\quad$}c@{$\quad$}||@{$\quad$}c@{$\quad$}|@{$\quad$}c@{$\quad$}|@{$\quad$}c@{$\quad$}|@{$\quad$}c@{$\quad$}|@{$\quad$}c@{$\quad$}}
$k$ & $1$ & $2$ & $3$ & $4$ & $\ldots$ \\
\hline
\hline
$\error(k,\varepsilon)$ & $\varepsilon^2$ & $\frac{2}{3}\varepsilon^3$ &
$\frac{1}{3}\varepsilon^4$ & $\frac{2}{15}\varepsilon^5$ & $\ldots$ \\
\hline

$\error_s(k,\varepsilon)$ &  $2 \varepsilon^2$ & $2 \varepsilon^3$ &
$\frac{17}{6}\varepsilon^4$ & $\frac{67}{30}\varepsilon^5$ & $\ldots$ \\
%\hline

%$\error_s(k,\varepsilon)$ & $\varepsilon^2$ & $\frac{1}{2}\varepsilon^3$ &
%$\frac{1}{3}\varepsilon^4$ & $\frac{1}{6}\varepsilon^5$ & $\ldots$ \\
%\hline

%$\error_p(k,\varepsilon)$ & $\frac{1}{2}\varepsilon^2$ & $\frac{2}{3}\varepsilon^3$ &
%$\frac{2}{6}\varepsilon^4$ & $\frac{2}{15}\varepsilon^5$ & $\ldots$ \\

%$\error_p(k,\varepsilon)$ & $\frac{1}{2}\varepsilon^2$ & $\frac{5}{6}\varepsilon^3$ &
%$\frac{1}{2}\varepsilon^4$ & $\frac{7}{30}\varepsilon^5$ & $\ldots$ \\
\end{tabular}
\end{center}

We can also compute, for a given precision~$\pi$, the value of~$\epsilon$ that
should be used in order to achieve an accuracy of~$\pi$ with~$\epsilon$-nets of
level~$k$.

\begin{Lemma}
\label{lem:kprecision}
To obtain a precision~$\pi$ with an $\varepsilon$-net of level~$k$,  we choose
$\varepsilon \approx \sqrt[k]{\frac{\pi}{T}}$, resulting in
$\frac{T}{\varepsilon}\approx T \sqrt[k]{\frac{T}{\pi}}$ 
%\frac{T^{1\frac{1}{k}}}{\sqrt[k]{\pi}}$ 
steps. 
\end{Lemma}

Unfortunately, the cost of computing $\epsilon$-nets of level~$k$ becomes
increasingly prohibitive as~$k$ increases. To see why, we first give a property
of the functions~$p_k$. Recall that~$p_2$ is a piecewise quadratic function. It
is not too difficult to see how this generalises to the approximations $p_k$.

\begin{lemma}\label{lem:piecepoly}
The approximation $p_k$ is piecewise polynomial with degree less than or equal
to $k$.
\end{lemma}

Although these functions are well-behaved in the sense that they are always
piecewise polynomial, the number of pieces can grow exponentially in the worst
case. The following lemma describes this bound.

\begin{Lemma}\label{lem:kswitchpoint}
If $p_{k-1}$ has $c$ pieces in the interval $[t - \epsilon, t]$, then $p_k$ has
at most $\frac{1}{2}\cdot c \cdot k \cdot |\locations| \cdot|\Sigma|^2$ pieces
in the interval $[t - \epsilon, t]$. 
\end{Lemma}

The upper bound given above is quite coarse, and we would be surprised if it
were found to be tight. Moreover, we do not believe that the number of pieces
will grow anywhere close to this bound in practice. This is because it is rare,
in our experience, for optimal strategies to change their decision many times
within a small time interval.

However, there is a more significant issue that makes $\epsilon$-nets become
impractical as~$k$ increases. In order to compute the approximation~$p_k$, we
must be able to compute the roots of polynomials with degree $k-1$. Since we can
only efficiently compute the roots of quadratic functions, and efficiently
approximate the roots of cubic functions, only the approximations~$p_3$
and~$p_4$ are realistically useful. 

Once again it is possible to provide a smart algorithm that uses sorting in
order to find the switching points in the functions~$p_3$ and~$p_4$, which gives
the following bounds on the cost of computing them.

\begin{Theorem}\label{theo:knets}
For a normed Markov $\mathcal{M}$ we can construct $\pi$ optimal memoryless
strategies for both players and determine the quality of these
strategies with precision $\pi$ in time $O(|\locations|^2 \cdot
\sqrt[3]{\frac{T}{\pi}}\cdot T \cdot |\Sigma|^4 \log |\Sigma|)$ when using triple $\varepsilon$-nets, and
in time $O(|\locations|^3 \cdot
\sqrt[4]{\frac{T}{\pi}}\cdot T \cdot |\Sigma|^6 \log |\Sigma|)$ when using quadruple $\varepsilon$-nets.
\end{Theorem}

It is not clear if triple and quadruple $\varepsilon$-nets will
only be of theoretical interest, or if they will be useful in practice. 
It should be noted that the worst case complexity bounds given by
Theorem~\ref{theo:knets} arise from the upper bound on the number of switching
points given in Lemma~\ref{lem:kswitchpoint}. Thus, if the number of switching
points that occur in practical examples is small, these techniques may become
more attractive. Our experiments in the following section give some evidence
that this may be true.

\section{Experimental Results and Conclusion}
\label{sec:experimental}

In order to test the practicability of our algorithms, we have implemented both
double and triple-$\epsilon$ nets. We evaluated these algorithms on two sets of
examples. Firstly, we tested our algorithms on  the Erlang-example (see
Figure~\ref{fig:ErlangModel}) presented
in~\cite{Buchholz+Hahn+Hermanns+Zhang/11/MC_CTMDPs}
and~\cite{Zhang+Neuh/10/ModelCheckingIMCs}. We chose to consider the same
parameters used by those papers: we consider maximal probability to reach
location $l_4$ from $l_1$ within 7 time units. Since this example is a CTMDP, we
were able to compare our results with the Markov Reward Model Checker
(MRMC)~\cite{Buchholz+Hahn+Hermanns+Zhang/11/MC_CTMDPs} implementation, which includes
an implementation of the techniques proposed by Buckholz and Schulz. 

We also tested our algorithms on continuous-time Markov games, where we used the
model depicted in Figure~\ref{fig:ScalableGameModel}, consisting of two chains
of locations $l_1,l_2,\dots,l_{100}$ and $l'_1,l'_2,\dots,l'_{100}$ that are
controlled by the maximising player and the minimising player, respectively.
This example is designed to produce a large number of switching points. In
every location~$l_i$ of the maximising player, there is the choice between the
short but slow route along the chain of maximising locations, and the slightly
longer route which uses the minimising player's locations. If very little
time remains, the maximising player prefers to take the slower actions, as
fewer transitions are required to reach the goal using these actions. The
maximiser also prefers these actions when a large amount of time remains.
However, between these two extremes, there is a time interval in which it is
advantageous for the maximising player to take the action with rate~3. A similar
situation occurs for the minimising player, and this leads to a large number of
points where the players change their strategy. 

\begin{figure}[t]
%\vspace{-.7cm}
\centering
%\large
%\vspace{-.3cm}
\begin{pspicture}[unit=0.9cm,showgrid=false](0,0.6)(10,3.3)
	\psset{arrowsize=5pt,nodesep=0pt,arrowlength=1,linewidth=1pt}
	\psline[]{->}(-.7,3.7)(-.25,3.25)
	\cnode[](0,3){10pt}{1}
	\rput(0,3){$l_1$}
	
	\definecolor{lightgrey}{rgb}{.8,.8,.8}
	\psframe[linestyle=none, fillstyle=solid,fillcolor=lightgrey,framearc=0.5](1.4,2.5)(8.6,3.5)
	\rput(5,2.2){Erlang(30,10)}
	
	\cnode[](2,3){10pt}{e1}
	
	\cnode[](4,3){10pt}{e2}
	
	\cnode[linecolor=lightgrey](6,3){10pt}{edot}
	\rput(6,3){\dots}
	
	\cnode[](8,3){10pt}{e30}
	
	\cnode[linewidth=0.8pt](10,3){10pt}{goal}
	\cnode[linewidth=0.8pt](10,3){8pt}{goalinner}
	\rput(10,3){$l_4$}
	
	\cnode[](5,1){10pt}{2}
	\rput(5,1){$l_3$}
	
	\cnode[](10,1){10pt}{5}
	\rput(10,1){$l_5$}
	
	\ncline[]{->}{1}{2}
	\ncline[]{->}{1}{e1}
	\ncline[]{->}{e1}{e2}
	\ncline[]{->}{e2}{edot}
	\ncline[]{->}{edot}{e30}
	\ncline[]{->}{1}{2}
	\ncline[]{->}{2}{goal}
	\ncline[]{->}{e30}{goal}
	\ncline[]{->}{2}{5}
	
	\rput(0.8,3.2){$a$,1}
	\rput(0.8,2.3){$b$,1}
	
	\rput(3,3.2){10}
	\rput(5,3.2){10}
	\rput(7,3.2){10}
	\rput(9,3.2){10}
	
	\rput(7.5,1.7){$a$,0.5}
	\rput(7.5,1.2){$a$,0.5}
\end{pspicture}
\caption{A CTMDP offering the choice between a long chain of fast transition and a slower path that looses some probability mass in $l_5$.}%
\label{fig:ErlangModel}%
\end{figure}
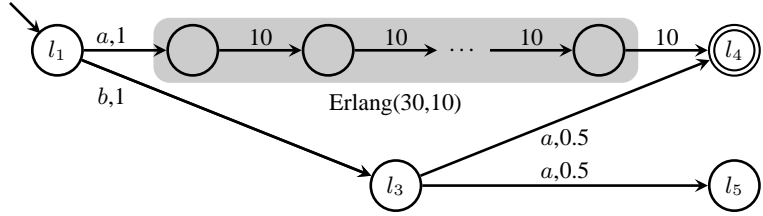

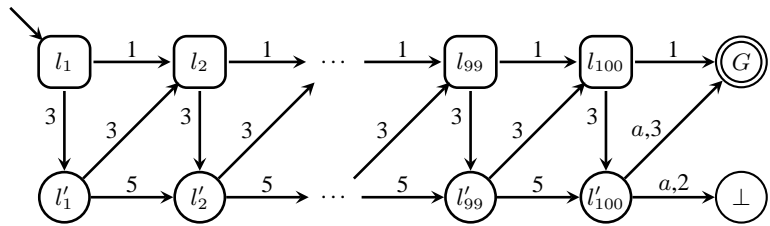
\begin{figure}[b]
%\vspace{-.7cm}
\centering
%\large
%\vspace{-.3cm}
\begin{pspicture}[unit=0.9cm,showgrid=false](0,0.6)(10,3.3)
	\psset{arrowsize=5pt,nodesep=0pt,arrowlength=1,linewidth=1pt}
	\psline[]{->}(-.8,3.8)(-.32,3.32)
	\cnode[linecolor=white](0,3){11pt}{max1}
	\psframe[fillstyle=none,framearc=0.5](-.4,2.6)(.4,3.4)
	\cnode[linecolor=white](2,3){11pt}{max2}
	\psframe[fillstyle=none,framearc=0.5](1.6,2.6)(2.4,3.4)
	\cnode[linecolor=white](4,3){11pt}{maxdots}
	\rput(4,3){\dots}
	\cnode[linecolor=white](6,3){11pt}{maxSecondToLast}
	\psframe[fillstyle=none,framearc=0.5](5.6,2.6)(6.4,3.4)
	\cnode[linecolor=white](8,3){11pt}{maxLast}
	\psframe[fillstyle=none,framearc=0.5](7.6,2.6)(8.4,3.4)
	
	\cnode[](0,1){10pt}{min1}
	\cnode[](2,1){10pt}{min2}
	\cnode[linecolor=white](4,1){10pt}{mindots}
	\rput(4,1){\dots}
	\cnode[](6,1){10pt}{minSecondToLast}
	\cnode[](8,1){10pt}{minLast}
	
	\cnode[linewidth=0.8pt](10,3){10pt}{goal}
	\cnode[linewidth=0.8pt](10,3){8pt}{goalinner}
	\rput(10,3){$G$}
	\cnode[linewidth=0.8pt](10,1){10pt}{deadend}
	\rput(10,1){$\bot$}
	
	\ncline[]{->}{max1}{max2}
	\ncline[]{->}{max2}{maxdots}
	\ncline[]{->}{maxdots}{maxSecondToLast}
	\ncline[]{->}{max1}{min1}
	\ncline[]{->}{max2}{min2}
	\ncline[]{->}{maxSecondToLast}{minSecondToLast}
	\ncline[]{->}{maxSecondToLast}{maxLast}
	\ncline[]{->}{maxLast}{minLast}
	\ncline[]{->}{maxLast}{goal}
	\ncline[]{->}{minLast}{goal}
	\ncline[]{->}{minLast}{deadend}
	
	\ncline[]{->}{min1}{min2}
	\ncline[]{->}{min2}{mindots}
	\ncline[]{->}{mindots}{minSecondToLast}
	\ncline[]{->}{minSecondToLast}{minLast}
	
	\ncline[]{->}{min1}{max2}
	\ncline[]{->}{min2}{maxdots}
	\ncline[]{->}{mindots}{maxSecondToLast}
	\ncline[]{->}{minSecondToLast}{maxLast}
	
	\rput(0,3){$l_1$}
	\rput(2,3){$l_2$}
	\rput(6,3){$l_{99}$}
	\rput(8,3){$l_{100}$}

	\rput(0,1){$l'_1$}
	\rput(2,1){$l'_2$}
	\rput(6,1){$l'_{99}$}
	\rput(8,1){$l'_{100}$}
	
	\rput(1,3.2){1}
	\rput(3,3.2){1}
	\rput(5,3.2){1}
	\rput(7,3.2){1}
	\rput(9,3.2){1}
	
	\rput(1,1.2){5}
	\rput(3,1.2){5}
	\rput(5,1.2){5}
	\rput(7,1.2){5}
	
	\rput(8.6,2){$a$,3}
	\rput(9,1.2){$a$,2}
	
	\rput(0.7,2){3}
	\rput(2.7,2){3}
	\rput(4.7,2){3}
	\rput(6.7,2){3}
	
	\rput(-0.2,2.2){3}
	\rput(1.8,2.2){3}
	\rput(5.8,2.2){3}
	\rput(7.8,2.2){3}
\end{pspicture}
\caption{A CTMG with many switching points.}%
\label{fig:ScalableGameModel}%
\end{figure}

\begin{table}[t]
\label{tbl:results}
\begin{center}
\begin{tabular}{@{~}c@{~} || @{~~}c@{~~} | @{~~}c@{~~}|@{~~}c@{~~} || @{~~}c@{~~}|@{~~}c@{~~}}

& \multicolumn{3}{c||@{~~}}{\bf Erlang model} & \multicolumn{2}{c}{\bf Game model} \\
precision \textbackslash\ method & MRMC \cite{Buchholz+Hahn+Hermanns+Zhang/11/MC_CTMDPs} & Double-nets & Triple-nets & Double-nets & Triple-nets \\
\hline
\hline

$10^{-4}$ &  
0.05 s & 0.04 s & 0.01 s & 0.34 s & 0.08 s \\
\hline

$10^{-5}$ &  
0.20 s & 0.10 s & 0.02 s & 1.04 s & 0.15 s \\
\hline

$10^{-6}$ & 
1.32 s & 0.32 s & 0.04 s & 3.29 s & 0.31 s \\
\hline

$10^{-7}$ & 
8 s & 0.98 s & 0.06 s & 10.45 s & 0.66 s \\
\hline

$10^{-8}$ &
475 s & 3.11 s & 0.14 s & 33.12 s & 1.42 s \\
\hline

$10^{-9}$ &
%502 s~${}^*$
--- & 9.91 s & 0.30 s & 106 s & 3.09 s \\
\hline

$10^{-10}$ &
%548 s~${}^*$ 
---& 31.24 s & 0.64 s & 339 s & 6.60 s \\
%\hline

\end{tabular}
\end{center}
\caption{Experimental evalutation of our algorithms.}
\end{table}

The results of our experiments are shown in Table~\ref{tbl:results}. The MRMC
implementation was unable to provide results for precisions beyond $1.86 \cdot
10^{-9}$. For the Erlang examples we found that, as the desired precision
increases, our algorithms draw further ahead of the current techniques. The most
interesting outcome of these experiments is the validation of triple
$\epsilon$-nets for practical use. While the worst case theoretical bounds
arising from Lemma~\ref{lem:kswitchpoint} indicated that the cost of computing
the approximation for each interval may become prohibitive, these results show
that the worst case does not seem to play a role in practice. In fact, we found
that the number of switching points summed over all intervals and locations
never exceeded 2 in this example.

%\paragraph{\bf Technical details. }
%All experiments were performed on a Notebook equipped with a Core i7 processor (M640@2.80 GHz), 4GB RAM (1333MHz FSB). 
%Our prototype is written in Java whereas MRMC is written in C. 
%MRMC was compiled with optimisation level {-O3} and it was executed inside a virtual machine on Ubuntu Linux. 
%Each value presented in the table above is the median of 5 individual runs. 
%Both implementations work with double precision floating-point numbers and do not exploit parallelism.

%\subsection{Game models}
%To the best of our knowledge, we provide the first practical algorithm to compute time-bounded reachabilities for CTMGs and in the following we give evidence for the practicability. 
%In the following, we provide further evaluation data for game models. 

Our results on Markov games demonstrate that our algorithms are capable of
solving non-trivially sized games in practice. Once again we find that triple
$\epsilon$-nets provide a substantial performance increase over double
$\epsilon$-nets, and that the worst case bounds given by
Lemma~\ref{lem:kswitchpoint} do not seem occur. Double $\epsilon$-nets found 297
points where the strategy changed during an interval, and triple $\epsilon$-nets
found 684 such points. Hence, the $|L||\Sigma|^2$ factor given in
Lemma~\ref{lem:kswitchpoint} does not seem to arise~here.

%During the evaluation we observed precisely the predicted behaviour. 
%We considered the points in time at which the strategy switches inside an interval (that is when we have a nontrivial estimator). 
%For triple nets, we have at least two such switching points for almost every location and 684 in total (in $>10000$ iterations), while for double nets the number reduces to 297. 
%This gives also evidence for our conjecture that the theoretical increase in switching points, which we discussed in Section~\ref{ssec:beyond}, does not actually occur in practice. 

\bibliographystyle{abbrv}
\bibliography{references}

\newpage
\appendix

\section{Proof of Lemma~\ref{lem:uniform}}

We first show how our algorithms can be used to solve uniform Markov games, and
then argue that this is sufficient to solve general Markov games. In order to
solve uniform Markov games with arbitrary uniformisation rate~$\lambda$, we will
define a corresponding normed Markov game in which time has been compressed by a
factor of $\lambda$. More precisely,  for each Markov game
$\M=(\locations,\locations_r,\locations_s,\Sigma,\ratematrix,\probabilitymatrix,\nu)$
with uniform transition rate $\lambda>0$, we define
$\M^{\|\cdot\|}=(\locations,\locations_r,\locations_s,\Sigma,\probabilitymatrix,\probabilitymatrix,\nu)$,
which is the Markov game that differs from $\M$ only in the rate matrix. In
particular, we replace $\ratematrix$ with $\probabilitymatrix =
\frac{1}{\lambda}\ratematrix$. The following lemma allows us to translate
solutions of $\M^{\|\cdot\|}$ to $\M$.

\begin{lemma}
\label{lem:boring}
%For every Markov game $\M$, if we have computed the optimal time-bounded reachability probability (or approximated it with precision $\pi$) for some time bound $T$, then we can compute optimal time bounded reachability probability  (or approximate it with precision $\pi$) for the time bound $\lambda T$.
%If we have constructed the respective strategies in $\M^{\|\cdot\|}$, we can construct the respective strategies in $\M$.
For every uniform Markov game~$\M$, if we have approximated the optimal
time-bounded reachability probabilities and strategies in $\M^{\|\cdot\|}$ with
some precision~$\pi$ for the time bound~$T$, then we can approximate optimal
time-bounded reachability probabilities and strategies in~$\M$ with precision~$\pi$ for the time bound $\lambda T$.
\end{lemma}
\begin{proof}
To prove this claim, we define the bijection $b:\mathcal S[\M^{\|\cdot\|}]
\rightarrow \mathcal S[\M]$ between schedulers of $\M^{\|\cdot\|}$ and $\M$ that
maps each scheduler $s\in  S[\M^{\|\cdot\|},T]$ to a scheduler $s' \in \mathcal
S[\M,\lambda T]$ with $s'(l,t) = s(l,\lambda t)$ for all $t\in [0,T]$. In other
words, we map each scheduler of $\M^{\|\cdot\|}$ to a scheduler of $\M$ in which
time has been stretched by a factor of $\lambda$. It is not too difficult to see
that the time-bounded reachability probability for time bound $\lambda T$ in
$\M$ under $s'=b(s)$ is equivalent to the time-bounded reachability probability
for time bound~$T$ for $\M^{\|\cdot\|}$ under~$s$. This bijection therefore
proves that the optimal time-bounded reachability probabilities are the same in
both games, and it also provides a procedure for translating approximately
optimal strategies of the game $\M^{\|\cdot\|}$ to the game $\M$. Since the
optimal reachability probabilities are the same in both games, an approximation
of the optimal reachability probability in $\M^{\|\cdot\|}$ with precision~$\pi$
must also be an approximation of the optimal reachability probability in
$\M^{\|\cdot\|}$ with precision~$\pi$.
%To translate our time complexity bounds
%to the general case, it suffices to replace each occurrence of~$T$ with
%$\lambda T$.
\end{proof}

In order to solve general Markov games we can first \emph{uniformise} them, and
then apply Lemma~\ref{lem:boring}. 
%For each continuous time Markov game $\M$, we define $\lambda_\M=\max_{l\in
%\locations} \ratematrix(l,a,\locations\setminus \{l\})$. To uniformise $\M$, we
%modify the rate matrix so that for all locations $l \in \locations$ and all
%actions $a \in \Sigma(l)$ we have $\ratematrix_{\mathcal
%U}(l,a,\locations)=\lambda_\M$. 
If
$\M=(\locations,\locations_r,\locations_s,\Sigma,\ratematrix,\probabilitymatrix,\nu)$
is a continuous-time Markov game, then we define the uniformisation of $\M$ as
$\unif(\M)=(\locations,\locations_r,\locations_s,\Sigma,\ratematrix',\probabilitymatrix,\nu)$,
where $\ratematrix'$ is defined as follows. If $\lambda=\max_{l\in
\locations} \max_{a \in \Sigma(l)} \ratematrix(l,a,\locations\setminus \{l\})$,
then we define, for every pair of locations $l, l' \in \locations$, and every
action $a \in \Sigma(l)$:
\begin{equation*}
\ratematrix'(l, a, l') = \begin{cases}
\ratematrix(l, a, l') & \text{if $l \ne l'$}, \\
\lambda - \ratematrix(l, a, L) & \text{if $l = l'$}.
\end{cases}
\end{equation*}

%The uniformisation of a Markov game $\M$ is the uniform Markov game $\mathcal U$ that is obtained from $\M$ by adjusting the rate matrix by adjusting $\ratematrix(l,a,l)$ for all locations $l \in \locations$ and all actions $a \in \Sigma(l)$ such that $\ratematrix_{\mathcal U}(l,a,\locations)=\lambda_\M$.

%For a Markov game $\M$ and its uniformisation $\mathcal U$, two simple properties of CTMGs~\cite{Rabe+Schewe/11/optimalSchedulersExist}:
%\begin{enumerate}
 %\item there are optimal timed-positional strategies for both players in all CTMGs, and thus in $\mathcal U$ and $\M$, with time-bounded reachability objective, and
 %\item the time bounded reachability for a timed-positional strategy or a pair of timed-positional strategies is not affected by the uniformisation (see Equation~\ref{eq:diffeqAlt}).
%\end{enumerate}
%%
%Consequently, it makes no difference whether we solve the problem for $\mathcal U$ or $\M$.

Previous work has noted that, for the class of late schedulers, the optimal
time-bounded reachability probabilities and schedulers in $\M$ are identical to
the optimal time-bounded reachability probabilities and schedulers in
$\unif(\M)$~\cite{Rabe+Schewe/11/optimalSchedulersExist}. To see why, note that
Equation~\eqref{eq:diffeqAlt} does not refer to the entry $\ratematrix'(l, a,
l)$, and therefore the modifications made to the rate matrix by uniformisation
can have no effect on the choice of optimal action.

\begin{lemma}
\cite{Rabe+Schewe/11/optimalSchedulersExist}
%For every Markov game $\M$ and its uniformisation $\mathcal U$, the time-bounded reachability probability, the time bounded reachability probability for a fixed timed-positional strategy of one player, and the time bounded reachability probability for fixed timed-positional strategies of both players coincide.
For every continuous-time Markov game $\M$, the optimal time-bounded
reachability probabilities and schedulers of $\M$ are identical to the optimal
time-bounded reachability probabilities and schedulers of $\unif(\M)$.
\end{lemma}

\section{Proof of Lemma~\ref{lem:individual}}

\begin{proof}
In order to prove this lemma, we will show that
$\|f(t-\varepsilon) - f_{p(t)}^t(t-\varepsilon)\|\leq \mu$. This implies the
claimed result, because by assumption we have $\|f_{p(t)}^t(t-\varepsilon) -
p(t-\varepsilon)\|\leq \nu$, and therefore the triangle inequality implies that
$\|f(t-\varepsilon) - p(t-\varepsilon)\|\leq \mu+\nu$.

We first prove that $\|f(t - \tau) - f^{t}_{f(t) + c}(t - \tau)\| = c$ for every
constant~$c$ and every $\tau \in [0, \epsilon]$. In other words, by increasing
the values of each location by~$c$ at time~$t$, we increase the values given
by~$f$ by~$c$ on the interval $[t - \epsilon, t]$. To see this, note
that by definition we have $\sum_{l'\in L} \mathbf{Q}(l,a,l') = 0$ for every
action~$a$, and therefore if we have $f'(t - \tau,l) = f(t - \tau,l) + c$ for
some time~$\tau \in [0, \epsilon]$, and every location~$l$, then we have:
\begin{equation*}
\sum_{l'\in L} \mathbf{Q}(l,a,l') f'(l, t - \tau) = \sum_{l'\in L}
\mathbf{Q}(l,a,l')(f(l, t - \tau) + c) = \sum_{l'\in L} \mathbf{Q}(l,a,l') f(l,
t - \tau).
\end{equation*}
We can then use this equality and Equation~\eqref{eq:2} to conclude that
$-\dot{f}_{f(t) + c}^t(l, t - \tau) = -\dot{f}(l, t - \tau)$ for every $\tau \in
[t - \epsilon, t]$, and therefore, we have $f^{t}_{f(t) + c}(t - \tau) - f(t -
\tau) = c$ for every $\tau \in [0, \epsilon]$. Since $\|f(t) - p(t)\| \leq \mu$,
it follows that $\|f(t-\varepsilon) - f_{p(t)}^t(t-\varepsilon)\|\leq \mu$ as
required. 
\end{proof}

\section{Proof of Lemma~\ref{extheo:single}}

\begin{lemma}
\label{lem:p1-01}
If $\epsilon \le 1$, then we have $p_1(l, t) \in [0, 1]$ for all $t \in [0,
T]$.
\end{lemma}
\begin{proof}
We will prove this by induction over the intervals $[t - \epsilon, t]$. The
base case is trivial since we have by definition that either $p_1(l, T) = 0$ or
$p_1(l, T) = 1$. Now suppose that $p_1(l, t) \in [0, 1]$ for some
$\epsilon$-interval $[t-\epsilon, t]$. We will prove that $p_1(l, t - \tau) \in
[0, 1]$ for all $\tau \in [0, \epsilon]$.

From the definition of~$p_1$, we know that $\dot{p_1}(l, t - \tau) = c_l^t =
\mathbf{R}(l,a_l^t,l') \cdot (p_1(l',t) - p_1(l, t))$ for all $\tau \in [0,
\epsilon]$. Therefore, since $\tau \le \epsilon \le 1$ we have:
\begin{align*}
p_1(l, t - \tau) &= p_1(l, t) + \tau \cdot \sum_{l'\in L}\mathbf{R}(l,a_l^t,l')
\cdot (p_1(l',t) - p_1(l, t))\\
&\le p_1(l, t) + \sum_{l'\in L}\mathbf{R}(l,a_l^t,l') \cdot (p_1(l',t) - p_1(l,
t))\\
&= \left(1-\sum_{l'\ne L}\mathbf{R}(l,a_l^t,l')\right) \cdot p_1(l, t) +
\sum_{l'\ne L}\mathbf{R}(l,a_l^t,l') \cdot p_1(l',t).
\end{align*}
Since we are considering normed Markov games, we have that $\sum_{l'\ne
L}\mathbf{R}(l,a_l^t,l') \le 1$, and therefore $p_1(l, t - \tau)$ is a weighted
average over the values $p_1(l', t)$ where $l' \in L$. From the inductive
hypothesis, we have that $p_1(l', t) \in [0, 1]$ for every $l' \in L$, and
therefore a weighted average over these values must also lie in $[0, 1]$.
\end{proof}

\begin{lemma}
\label{lem2:p2}
If $\epsilon \le 1$ then we have $-\dot{f}^t_{p_1(t)}(l, t - \tau) \in [-1, 1]$
for every $\tau \in [0, \epsilon]$.
\end{lemma}
\begin{proof}
Lemma~\ref{lem:p1-01} implies that $f^t_{p_1(t)}(l, t) = p_1(l, t) \in [0, 1]$
for all $l \in L$. Since the system of differential equations given
by~\eqref{eq:2} gives optimal reachability probabilities under the assumption
that $f^t_{p_1(t)}(l, t) = p_1(l, t)$, we must have that $f^t_{p_1(t)}(l, t -
\tau) \in [0, 1]$ for all $\tau \in [0, \epsilon]$. 

We first prove that $-\dot{f}^t_{p_1(t)}(l, t - \tau) \le 1$. We will prove this
for the reachability player, the proof for the safety player is analogous. By
definition we have:
\begin{equation*}
-\dot{f}_{x}^t(l,t-\tau)=\max_{a\in\Sigma(l)}\sum_{l'\in L} \mathbf{R}(l,a,l')
(f_{p_1(t)}^t(l',t-\tau) - f_{p_1(t)}^t(l,t-\tau)).
\end{equation*}
Since we have shown that $f_{p_1(t)}^t(l',t-\tau) \in [0, 1]$ for all $l$, and we have
$\sum_{l'\in L} \mathbf{R}(l,a,l') = 1$ for every action $a$ in a normed Markov
game, we obtain:
\begin{equation*}
\max_{a\in\Sigma(l)}\sum_{l'\in L} \mathbf{R}(l,a,l') (f_{p_1(t)}^t(l',t-\tau) -
f_{p_1(t)}^t(l,t-\tau))  \le \max_{a\in\Sigma(l)}\sum_{l'\in L} \mathbf{R}(l,a,l') (1 -
0) = 1
\end{equation*}
To prove that $-\dot{f}^t_{p_1(t)}(l, t - \tau) \ge -1$ we use a similar
argument: 
\begin{equation*}
\max_{a\in\Sigma(l)}\sum_{l'\in L} \mathbf{R}(l,a,l') (f_{p_1(t)}^t(l',t-\tau) -
f_{p_1(t)}^t(l,t-\tau))  \ge \max_{a\in\Sigma(l)}\sum_{l'\in L} \mathbf{R}(l,a,l') (0 -
1) = -1 
\end{equation*}
Therefore we have $-\dot{f}^t_{p_1(t)}(l, t - \tau) \in [-1, 1]$.
\end{proof}

We can now provide a proof for Lemma~\ref{extheo:single}.

\begin{proof}
Lemma~\ref{lem2:p2} implies that $-\dot{f}^t_{p_1(t)}(l, t - \tau) \in [-1, 1]$
for every $\tau \in [0, \epsilon]$. Since the rate of change of $f^t_{p_1(t)}$
is in the range $[-1, 1]$, we know that $f^t_{p_1(t)}$ can change by at
most~$\tau$ in the interval $[t - \tau, t]$. We also know
that $f^t_{p_1(t)}(l, t) = p_1(l, t)$, and therefore we must have the following
property:
\begin{equation}
\label{keyeqn}
\| f^t_{p_1(t)}(l, t - \tau) - p_1(l, t)\| \le \tau.
\end{equation}

The key step in this proof is to show that $\| \dot{f}^t_{p_1(t)}(l, t - \tau) -
\dot{p}_1(l, t - \tau) \| \le 2 \cdot \tau$ for all $\tau \in [0, \epsilon]$.
Note that by definition we have $\dot{p}_1(l, t - \tau) = \dot{p}_1(l, t)$ for
all $\tau \in [0, \epsilon]$, and so it suffices to prove that $\|
\dot{f}^t_{p_1(t)}(l, t - \tau) - \dot{p}_1(l, t) \| \le 2 \cdot \tau$.

Suppose that $l$ is a location for the reachability player, let $a^t_l$ be the
optimal action at time~$t$, and let $a^{t-\tau}_l$ be the optimal action at
$t-\tau$. We have the following:
\begin{align*}
-\dot{p}_1(l, t) - 2 \cdot \tau &= 
\sum_{l'\in L} \mathbf{R}(l,a^t_l,l') (p_1(l', t) - p_1(l, t)) - 2 \cdot \tau \\
&\le \sum_{l'\in L} \mathbf{R}(l,a^t_l,l') (f^t_{p_1(t)}(l', t - \tau) -
f^t_{p_1(t)}(l, t - \tau)) \\
&\le \sum_{l'\in L} \mathbf{R}(l,a^{t-\tau}_l,l') (f^t_{p_1(t)}(l', t - \tau) -
f^t_{p_1(t)}(l, t - \tau)) =  -\dot{f}^t_{p_1(t)}(l, t - \tau)
\end{align*}
The first equality is the definition of $-\dot{p}_1(l, t)$. The first inequality
follows from Equation~\eqref{keyeqn} and the fact that $\mathbf{R}(l,a,l') = 1$.
The second inequality follows from the fact that $a^{t-\tau}_l$ is an optimal
action at time $t-\tau$, and the final equality is the definition of
$-\dot{f}^t_{p_1(t)}(l, t - \tau)$. Using the same techniques in a different
order gives:
\begin{align*}
-\dot{f}^t_{p_1(t)}(l, t - \tau) &= \sum_{l'\in L} \mathbf{R}(l,a^{t-\tau}_l,l')
(f^t_{p_1(t)}(l', t - \tau) - f^t_{p_1(t)}(l, t - \tau))  \\
&\le \sum_{l'\in L} \mathbf{R}(l,a^{t-\tau}_l,l') (p_1(l', t) - p_1(l, t)) + 2 \cdot \tau \\
&\le \sum_{l'\in L} \mathbf{R}(l,a^{t}_l,l') (p_1(l', t) - p_1(l, t
)) + 2 \cdot \tau = -\dot{p}_1(l, t) + 2 \cdot \tau
\end{align*}

To prove the claim for a location~$l$ belonging to the safety player, we use
the same arguments, but in reverse order. That is, we have:
\begin{align*}
-\dot{p}_1(l, t) - 2 \cdot \tau &= \sum_{l'\in L} \mathbf{R}(l,a^{t}_l,l') (p_1(l', t) - p_1(l, t)) - 2 \cdot \tau \\ 
&\le \sum_{l'\in L} \mathbf{R}(l,a^{t - \tau}_l,l') (p_1(l', t) - p_1(l,
t)) - 2 \cdot \tau \\ 
&\le \sum_{l'\in L} \mathbf{R}(l,a^{t-\tau}_l,l') (f^t_{p_1(t)}(l', t - \tau) -
f^t_{p_1(t)}(l, t - \tau)) = -\dot{f}^t_{p_1(t)}(l, t - \tau)
\end{align*}
We also have:
\begin{align*}
-\dot{f}^t_{p_1(t)}(l, t - \tau) &= \sum_{l'\in L} \mathbf{R}(l,a^{t-\tau}_l,l')
(f^t_{p_1(t)}(l', t - \tau) - f^t_{p_1(t)}(l, t - \tau)) \\
&\le \sum_{l'\in L} \mathbf{R}(l,a^{t}_l,l') (f^t_{p_1(t)}(l', t - \tau) -
f^t_{p_1(t)}(l, t - \tau)) \\
& \le \sum_{l'\in L} \mathbf{R}(l,a^t_l,l')
(p_1(l', t) - p_1(l, t)) + 2 \cdot \tau = -\dot{p}_1(l, t) + 2 \cdot \tau 
\end{align*}
Therefore, we have shown that $\| \dot{f}^t_{p_1(t)}(l, t - \tau) -
\dot{p}_1(l, t - \tau) \| \le 2 \cdot \tau$ for all $\tau \in [0, \epsilon]$ and
for every $l \in L$.

We now complete the proof by arguing that $\|f_{p_1(t)}^t(t-\tau)-p_1(t-\tau)\|
\le \tau^2$. Let the $d(l, t- \tau) := \|f_{p_1(t)}^t(t-\tau)-p_1(t-\tau)\|$ for
every $\tau \in [0, \epsilon]$. Our arguments so far imply:
\begin{equation*}
\dot{d}(l, t- \tau) \le \|\dot{f}_{p_1(t)}^t(t-\tau)-\dot{p}_1(t-\tau)\| \le
2 \cdot \tau.
\end{equation*}
Therefore, we have:
\begin{equation*}
d(l, t- \tau) \le \int_0^\tau 2 \cdot \tau d\tau = \tau^2.
\end{equation*}
This allows us to conclude that $\error(1, \epsilon) :=
\|f_{p_1(t)}^t(t-\epsilon)-p_1(t-\epsilon)\| \le \epsilon^2$.
\end{proof}

\section{Proof of Lemma~\ref{lem:single}}

%We can now show that the strategy for the reachability player comes within an
%$\epsilon^2$ margin of the true value given by $f_{p_1(t)}^t$. The key idea in
%this proof is that, if single $\varepsilon$-nets are applied to the
%function~$g_1$, then we obtain the same approximation~$p_1$ that is obtained when
%single $\varepsilon$-nets are applied to $f_{p_1(t)}^t$. Since $p_1$ is an
%approximation for $g_1$, we have the following lemma.

%Since~$p_1$ is also an approximation for $f_{p_1(t)}^t$, we have that
%$\|f_{p_1(t)}^t(t-\tau)-p_1(t-\tau)\| \le \frac{1}{2}\varepsilon^2$. This
%gives the following bound for the difference between~$g_1$ and~$f_{p_1(t)}^t$.

We begin by proving the following auxiliary lemma, which shows that the
difference between $p_1$ and $g_1$ is bounded by $\varepsilon^2$.
\begin{lemma}
\label{lem:errorp}
We have $\|g_1(t-\varepsilon)-p_1(t-\varepsilon)\| \le \varepsilon^2$.
\end{lemma}
\begin{proof}
Suppose that we apply single-$\epsilon$ nets to approximate the solution of the
system of differential equations~$g_1$ over the interval $[t-\epsilon, t]$ to
obtain an approximation $p^{g}_{1}$. To do this, we select for each location $l \in
L_s$ an action~$a$ that satisfies:
\begin{equation*}
a \in \argopttwo_{a\in\Sigma(l)} \sum_{l'\in L}\mathbf{Q}(l,a,l') \cdot g_1(l',t).
\end{equation*}
Since $g_1(l, t) = f^t_{p_1(t)}(l, t) = p_1(l, t)$ for every location~$l$, we
have that $a = a_l^t$, where $a_l^t$ is the action chosen by $p_1$ at $l$. In
other words, the approximations~$p_1$ and~$p^{g}_1$ choose the same actions for
every location in $L_s$. Therefore, for all locations $l \in L$, we have $c_l^t
= \sum_{l'\in L}\mathbf{Q}(l,a_l^t,l') \cdot p^{g}_1(l',t) = \sum_{l'\in
L}\mathbf{Q}(l,a_l^t,l') \cdot p_1(l',t)$, which implies that for every time
$\tau \in [0, \epsilon]$ we have:
\begin{align*}
p^{g}_1(l, t - \tau) = p_1^{g}(l, t) + \tau \cdot c_l^t = p_1(l, t) + \tau
\cdot c_l^t = p_1(l, t - \tau).
\end{align*}
That is, the approximations~$p_1$ and $p^g_1$ are identical.

Note that the system of differential equations~$g_1$ describes a continuous-time
Markov game in which some actions for the reachability player have been removed.
Since~$g_1$ describes a CTMG, we can apply Lemma~\ref{extheo:single} to obtain
$\|g_1(t-\tau)-p^g_1(t-\tau)\| \le \epsilon^2$. Since $p_1(t -
\epsilon) = p_1^g(t - \epsilon)$, we can conclude that $\|g_1(t-\tau) -
p_1(t-\tau)\| \le \epsilon^2$.
\end{proof}

Lemma~\ref{lem:single} now follows from Lemma~\ref{lem:errorp} and
Lemma~\ref{extheo:single}.

\section{Proof of Theorem~\ref{thm:singlenets}}

\begin{proof}
As we have argued in the main text, in order to guarantee a precision of $\pi$,
it suffices to choose $\epsilon = \frac{\pi}{T}$, which gives
$\frac{T^2}{\pi}$ many intervals $[t - \epsilon, t]$ for which $p_1$ must be
computed. It is clear that, for each interval, the approximation $p_1$ can be
computed in $O(\mathcal M)$ time, and therefore, the total running time will
be $O(|\M| \cdot T \cdot \frac{T}{\pi})$.
\end{proof}

\section{Proof of Lemma~\ref{lem:double}}

\begin{proof}
We begin by considering the system of differential equations that define $p_2$,
as given in Equation~\eqref{eq:approx2}:
\begin{equation*}
-\dot{p}_2(l,\tau) = \opttwo_{a\in\Sigma(l)}
\sum_{l'\in\locations} \ratematrix(l,a,l') \cdot \left(p_1(l',\tau) -
p_1(l,\tau)\right)\quad \forall l\in\locations.
\end{equation*}
The error bounds given by Lemma~\ref{extheo:single} imply that $\| p_1(t -
\tau)-f_{p_2(t)}^t(t - \tau) \| \le \tau^2$ for every $\tau \in [0,
\epsilon]$. Therefore, for every pair of locations $l,l' \in L$ and every $\tau
\in [t - \epsilon, t]$ we have:
\begin{equation*}
\|(p_1(l',t - \tau) - p_1(l, t - \tau)) - (f_{p_2(t)}^t(l', t -
\tau)-f_{p_2(t)}^t(l, t - \tau))\| \le 2 \cdot \tau^2.
\end{equation*}
Since we are dealing with normed Markov games, we have $\sum_{l' \in
\locations}\ratematrix(l, a, l') = 1$ for every location $l \in L$ and every
action $a \in A(l)$. Therefore, we also have for every action $a$:
\begin{multline*}
\| \sum_{l'\in\locations} \ratematrix(l,a,l') \cdot \left(p_1(l',t - \tau)
- p_1(l,t - \tau)\right) \\ - \sum_{l'\in\locations} \ratematrix(l,a,l') \cdot
(f_{p_2(t)}^t(l', t - \tau)-f_{p_2(t)}^t(l, t - \tau) \| \le  2 \cdot \tau^2.
\end{multline*}
This implies that $\|\dot{p}_2(l,t - \tau) - \dot{f}^t_{p_2(t)}(l, t - \tau)\|
\le 2 \cdot \tau^2$.

We can obtain the claimed result by integrating over this difference:
\begin{equation*}
\|p_2(l,t - \tau) - f^t_{p_2(t)}(l, t - \tau)\| = \int_0^\tau
\|\dot{p}_2(l,t - \tau) - \dot{f}^t_{p_2(t)}(l, t - \tau)\| \le
\frac{2}{3}\tau^3.
\end{equation*}
Therefore, the total amount of error incurred by~$p_2$ in the interval $[t -
\epsilon, t]$ is at most $\frac{2}{3}\epsilon^3$.
\end{proof}

\section{Proof of Lemma~\ref{lem:doubleConcrete}}

To begin, we prove an auxiliary lemma, that will be used throughout the rest of
the proof.

\begin{lemma}
\label{lem:dcuseful}
Let $f$ and $g$ be two functions such that $\|f(t - \tau) - g(t - \tau)\| \le
c \cdot \tau^k$. If $a_f$ is an action that maximises (resp. minimises)
\begin{equation}
\label{eqn:fclub}
\opttwo_{a\in\Sigma(l)} \sum_{l'\in\locations} \mathbf{Q}(l,a,l') \cdot f(l',t
- \tau), 
\end{equation}
and $a_g$ is is an action that maximises (resp. minimises)
\begin{equation}
\label{eqn:gclub}
\opttwo_{a\in\Sigma(l)} \sum_{l'\in\locations} \mathbf{Q}(l,a,l') \cdot g(l',t
- \tau), 
\end{equation}
then we have:
\begin{equation*}
\|\sum_{l'\in\locations} \mathbf{R}(l,a^g,l') \cdot g(t - \tau) -
\sum_{l'\in\locations} \mathbf{R}(l,a^f,l') \cdot f(t - \tau) \| \le 3 \cdot c
\cdot \tau^k.
\end{equation*}
\end{lemma}
\begin{proof}
We will provide a proof for the case where the equations must be maximises, the
proof for the minimisation case is identical. We begin by noting that the
property $\|f(t - \tau) - g(t - \tau)\| \le c \cdot \tau^k$, and the fact that
we consider only normed Markov games imply that, for every action~$a$ we have:
\begin{equation}
\label{eqn:protoss}
\|\sum_{l'\in\locations} \mathbf{Q}(l,a,l') \cdot f(l',t -
\tau) - \sum_{l'\in\locations} \mathbf{Q}(l,a,l') \cdot g(l',t -
\tau) \| \le 2 \cdot c \cdot \tau^k. 
\end{equation}
We use this to claim that the following inequality holds:
\begin{equation}
\label{eqn:zerg}
\|\sum_{l'\in\locations} \mathbf{Q}(l,a^g,l') \cdot g(l',t -
\tau) - \sum_{l'\in\locations} \mathbf{Q}(l,a^f,l') \cdot f(l',t -
\tau) \| \le 2 \cdot c \cdot \tau^k.
\end{equation}
To see why, suppose that 
\begin{equation*}
\sum_{l'\in\locations} \mathbf{Q}(l,a^g,l') \cdot
g(l',t - \tau) > \sum_{l'\in\locations} \mathbf{Q}(l,a^f,l') \cdot f(l',t -
\tau) + 2 \cdot c \cdot \tau^k.
\end{equation*} 
Then we could invoke Equation~\eqref{eqn:protoss} to argue that
$\sum_{l'\in\locations} \mathbf{Q}(l,a^g,l') \cdot f(l',t - \tau) >
\sum_{l'\in\locations} \mathbf{Q}(l,a^f,l') \cdot f(l',t - \tau)$, which
contradicts the fact that~$a^f$ achieves the maximum in
Equation~\eqref{eqn:fclub}. Similarly, if $\sum_{l'\in\locations}
\mathbf{Q}(l,a^f,l') \cdot f(l',t - \tau) > \sum_{l'\in\locations}
\mathbf{Q}(l,a^g,l') \cdot g(l',t - \tau) + 2 \cdot c \cdot \tau^k$, then we can
invoke Equation~\eqref{eqn:protoss} to argue that $a^g$ does not achieve the
maximum in Equation~\eqref{eqn:gclub}. Therefore, Equation~\eqref{eqn:zerg} must
hold.

Now, to finish the proof, we apply the fact that $\|f(t - \tau) - g(t -
\tau)\| \le c \cdot \tau^k$ to Equation~\eqref{eqn:zerg} to obtain:
\begin{equation*}
%\label{eqn:tatsubaki}
\| \sum_{l'\in\locations} \ratematrix(l, a^g, l')g(l', t - \tau) -
\sum_{l'\in\locations} \ratematrix(l, a^f, l') f(l',t - \tau) \| \le
3 \cdot c \cdot \tau^k
\end{equation*}
This complete the proof.
\end{proof}

To prove Lemma~\ref{lem:doubleConcrete} we will consider the following class of
strategies: play the action chosen by~$p_2$ for the first~$k$ transitions, and
then play the action chosen by~$p_1$ for the remainder of the interval. We will
denote the reachability probability obtained by this strategy as $g_2^k$, and we
will denote the error of this strategy as $\error^{k}_s(2, \epsilon) :=
\|g_2^k(t - \epsilon) - f^t_{p_2(t)}(t - \epsilon)\|$. Clearly, as~$k$
approaches infinity, we have that~$g_2^k$ approaches~$g_2$, and $\error^{k}_s(2,
\epsilon)$ approaches $\error_s(2, \epsilon)$. Therefore, in order to prove
Lemma~\ref{lem:doubleConcrete}, we will show that $\error^{k}_s(2,
\epsilon) \le 2 \cdot \epsilon^3$ for all~$k$.

We will prove error bounds on~$g_2^k$ by induction. The following lemma
considers the base case, where~$k = 1$. In other words, it considers the
strategy that plays the action chosen by~$p_2$ for the first transition, and
then plays the action chosen by~$p_1$ for the rest of the interval.

\begin{lemma}
\label{lem:dcbase}
If $\varepsilon \le 1$, then we have $\error^{1}_s(2, \epsilon) \le 2
\cdot \epsilon^3$.
\end{lemma}
\begin{proof}
Suppose that the first discrete transition occurs at time $t - \tau$, where
$\tau \in [0, \epsilon]$. Let~$l$ be a location belonging to the reachability
player, and let $a^{t-\tau}_l$ be the action that maximises~$p_2$ at time $t -
\tau$. By definition, we know that the probability of moving to a location~$l'$
is given by $\ratematrix(l, a_l^{t-\tau}, l')$, and we know that the
time-bounded reachability probabilities for each state $l'$ are given by
$g_1(l', t-\tau)$. Therefore, the outcome of choosing $a_l^{t-\tau}$ at time $t-
\tau$ is $\sum_{l'\in\locations} \ratematrix(l, a_l^{t-\tau}, l')g_1(l', t -
\tau)$. If $a^*$ is an action that would be chosen by $f^t_{p_2(t)}$ at time
$t-\tau$, then we have the following bounds:
\begin{align*}
&\| \sum_{l'\in\locations} \ratematrix(l, a_l^{t-\tau}, l')g_1(l', t - \tau) -
\sum_{l'\in\locations} \ratematrix(l, a^*, l') f^t_{p_2(t)}(l',t - \tau) \| \\
\le &\| \sum_{l'\in\locations} \ratematrix(l, a_l^{t-\tau}, l')p_1(l', t - \tau) -
\sum_{l'\in\locations} \ratematrix(l, a^*, l') f^t_{p_2(t)}(l',t - \tau) \| +
\tau^2  \\
\le& 4 \cdot \tau^2
\end{align*}
The first inequality follows from Lemma~\ref{lem:errorp}, and the second
inequality follows from Lemma~\ref{lem:dcuseful}.

Now suppose that~$l$ is a location belonging to the safety player. Since the
reachability player will follow $p_1$ during the interval $[t-\tau, t]$, we know
that the safety player will choose an action $a_g$ that minimises:
\begin{equation*}
\min_{a\in\Sigma(l)} \sum_{l'\in\locations} \mathbf{Q}(l,a,l') \cdot
g_1(l',t - \tau). 
\end{equation*}
If $a^*$ is the action chosen by $f$ at time $t - \tau$, then
Lemma~\ref{lem:single} and Lemma~\ref{lem:dcuseful} imply:
\begin{equation*}
\| \sum_{l'\in\locations} \ratematrix(l, a^g, l')g_1(l', t - \tau) -
\sum_{l'\in\locations} \ratematrix(l, a^*, l') f^t_{p_2(t)}(l',t - \tau) \| \le
6 \cdot \tau^2 \\
\end{equation*}

So far we have proved that the total amount of error made by $g_2^1$ when the
first transition occurs at time $t-\tau$ is at most $6 \cdot \tau^2$. To obtain
error bounds for $g_2^1$ over the entire interval $[t-\epsilon, t]$, we consider
the probability that the first transition actually occurs at time $t-\tau$:
\begin{equation*}
\error^{1}_s(2,\varepsilon) \leq \int_0^\varepsilon e^{\tau-\varepsilon} 6\tau^2
d\tau \leq  \int_0^\varepsilon 6\tau^2 d\tau= 2 \cdot \varepsilon^3.
\end{equation*}
This completes the proof.
\end{proof}

We now prove the inductive step, by considering~$g_2^k$. This is the strategy
that follows the action chosen by~$p_2$ for the first~$k$ transitions, and then
follows $p_1$ for the rest of the interval. 

\begin{lemma}
If $\error^{k}_s(2, \epsilon) \le 2 \cdot \epsilon^3$ for some~$k$, then
$\error^{k+1}_s(2, \epsilon) \le 2 \cdot \epsilon^3$.
\end{lemma}
\begin{proof}
The structure of this proof is similar to the proof of Lemma~\ref{lem:dcbase},
however, we must account for the fact that~$g_2^{k+1}$ follows~$g_2^k$ after the
first transition rather than~$g_1$.

Suppose that we play the strategy for $g_2^{k+1}$, and that the first discrete
transition occurs at time $t - \tau$, where $\tau \in [0, \epsilon]$. Let~$l$ be
a location belonging to the reachability player, and let $a^{t-\tau}_l$ be the
action that maximises~$p_2$ at time $t - \tau$. If $a^*$ is an action that would
be chosen by $f^t_{p_2(t)}$ at time $t-\tau$, then we have the following bounds:
\begin{align*}
&\| \sum_{l'\in\locations} \ratematrix(l, a_l^{t-\tau}, l')g_2^{k}(l', t - \tau)
- \sum_{l'\in\locations} \ratematrix(l, a^*, l') f^t_{p_2(t)}(l',t - \tau) \| \\
\le &\| \sum_{l'\in\locations} \ratematrix(l, a_l^{t-\tau}, l')p_1(l', t - \tau)
- \sum_{l'\in\locations} \ratematrix(l, a^*, l') f^t_{p_2(t)}(l',t - \tau) \| +
\tau^2 + 2 \cdot \tau^3 \\ 
\le& 4 \cdot \tau^2 + 2 \cdot \tau^3 \le 6 \cdot \tau^2
\end{align*}
The first inequality follows from the inductive hypothesis, which gives bounds
on how far $g_2^k$ is from $f^t_{p_2(t)}$, and from Lemma~\ref{extheo:single},
which gives bounds on how far~$f^t_{p_2(t)}$ is from~$p_1$. The second
inequality follows from Lemma~\ref{extheo:single} and Lemma~\ref{lem:dcuseful},
and the final inequality follows from the fact that $\tau \le 1$.

Now suppose that the location~$l$ belongs to the safety player. Let $a_g$ be an
action that minimises:
\begin{equation*}
\min_{a\in\Sigma(l)} \sum_{l'\in\locations} \mathbf{Q}(l,a,l') \cdot
g_2^k(l',t - \tau). 
\end{equation*}
If $a^*$ is the action chosen by~$f$ at time $t - \tau$, then
Lemma~\ref{lem:single} and Lemma~\ref{lem:dcuseful} imply:
\begin{equation*}
\| \sum_{l'\in\locations} \ratematrix(l, a^g, l')g_2^k(l', t - \tau) -
\sum_{l'\in\locations} \ratematrix(l, a^*, l') f^t_{p_2(t)}(l',t - \tau) \| \le
6 \cdot \tau^3 \le 6 \cdot \tau^2\\
\end{equation*}
The first inequality follows from the inductive hypothesis and
Lemma~\ref{lem:dcuseful}, and the second inequality follows from the fact that
$\tau \le 1$.

To obtain error bounds for $g_2^{k+1}$ over the entire interval $[t-\epsilon,
t]$, we consider the probability that the first transition actually occurs at
time $t-\tau$:
\begin{equation*}
\error^{k+1}_s(2,\varepsilon) \leq \int_0^\varepsilon e^{\tau-\varepsilon} 6
\cdot \tau^2 \;d\tau \leq  \int_0^\varepsilon 6\tau^2 \; d\tau= 2 \cdot
\varepsilon^3.
\end{equation*}
This completes the proof.
\end{proof}

\section{Proof of Lemma~\ref{lem:sortdouble}}

We give the algorithm for the reachability player. The algorithm for the safety
player is symmetric. For every location $l \in L$, and time point $\tau \in [0,
\epsilon]$, we define the \emph{quality} of an action $a$ as:
\begin{equation*}
q_l^{\tau}(a) := \sum_{l'\in L} \mathbf{Q}(l,a,l') p_1^t(l',t -\tau).
\end{equation*}
We also define an operator that compares the quality of two actions. For two
actions $a_1$ and $a_2$, we have $a_1 \preceq^{\tau}_l a_2$ if and only if
$q_l^{\tau}(a_1) \leq  q_l^{\tau}(a_2)$, and we have $a_1 \prec^{\tau}_l a_2$ if
and only if $q_l^{\tau}(a_1) <  q_l^{\tau}(a_2)$.

\begin{algorithm}
\begin{algorithmic}
\STATE Sort the actions into a list $\langle a_1, a_2, \dots, a_m \rangle$ such
that $a_i \preceq^{0}_l a_{i+1}$ for all $i$.
\STATE $O := \langle (a_1, 0) \rangle$.
\FOR{$i := 2$ to $m$}
	\STATE $(a, \tau) := $ the last element in $O$.
	\IF{$a \prec^{\epsilon}_l a_i$}
		\WHILE{\textbf{true}}
			\STATE $x :=$ the point at which $q_l^{x}(a) = q_l^{x}(a_i)$.
			\IF{$x \ge \tau$}
				\STATE Add $(a_i, x)$ to the end of $O$.
				\STATE \textbf{break}
			\ELSE
				\STATE Remove $(a, \tau)$ from $O$.
				\STATE $(a, \tau) := $ the last element in $O$.
			\ENDIF
		\ENDWHILE
	\ENDIF
\ENDFOR
\STATE \textbf{return}~$O$.
\end{algorithmic}
\caption{BestActions}
\label{alg:quaddamage}
\end{algorithm}

Algorithm~\ref{alg:quaddamage} shows the key component of our algorithm for
computing the approximation $p_2$ during the interval $[t-\epsilon, t]$. The
algorithms outputs a list~$O$ containing pairs $(a, \tau)$, where~$a$ is an
action and~$\tau$ is a point in time, which represents the optimal actions
during the interval $[t-\epsilon, t]$: if the algorithm outputs the list $O =
\langle (a_1, \tau_1), (a_2, \tau_2), \dots, (a_n, \tau_n)\rangle$, then $a_1$
maximises Equation~\eqref{eq:approx2} for the interval $[t - \tau_2, t -
\tau_1]$, $a_2$ maximises Equation~\eqref{eq:approx2} for the interval $t -
\tau_3, t - \tau_2]$, and so on.

The algorithm computes~$O$ as follows. It begins by sorting the actions
according to their quality at time~$t$. Since $a_1$ maximises the quality at
time~$t$, we know that~$a_1$ is chosen by Equation~\eqref{eq:approx2} at
time~$t$. Therefore, the algorithm initialises~$O$ by assuming that~$a_1$
maximises Equation~\eqref{eq:approx2} for the entire interval $[t - \epsilon,
t]$. The algorithm then proceeds by iterating through the actions $\langle a_2,
a_3, \dots a_m \rangle$. 

We will prove the following invariant on the outer loop of the algorithm: if the
first~$i$ actions have been processed, then the list~$O$ gives the solution to:
\begin{equation}
\label{eqn:mode2}
-\dot{p}_2(l,\tau,i) = \max_{a\in\langle a_1, a_2, \dots, a_i \rangle}
\sum_{l'\in\locations} \ratematrix(l,a,l') \cdot \left(p_1(l',\tau) -
p_1(l,\tau)\right).
\end{equation}
In other words, the list~$O$ would be a solution to
Equation~\eqref{eq:approx2} if the actions $\langle a_{i+1}, a_{i+2}, \dots
a_m\rangle$ did not exist. Clearly, when $i = m$ the list~$O$ will actually be a
solution to Equation~\eqref{eq:approx2}.

We will prove this invariant by induction. The base case is trivially true,
because when $i = 1$ the maximum in Equation~\eqref{eqn:mode2} considers
only~$a_1$, and therefore~$a_1$ is optimal throughout the interval $[t-\epsilon,
t]$. We now prove the inductive step. Assume that~$O$ is a solution to
Equation~\eqref{eqn:mode2} for~$i-1$. We must show that
Algorithm~\ref{alg:quaddamage} correctly computes~$O$ for $i$. Let us consider
the operations that Algorithm~\ref{alg:quaddamage} performs on the action
$a_{i}$. It compares $a_{i}$ with the pair $(a, \tau)$, which is the final
pair in~$O$, and one of three actions is performed:
\begin{itemize}
\item If $a_i \prec^{\epsilon}_l a$, then the algorithm ignores $a_{i}$. This is
because we have $a_{i} \prec^{0}_l a_1$, which means that $a_{i}$ is worse
than~$a_1$ at time~$t$, and we have $a_{i} \prec^{\epsilon}_l a$, which implies
that $a_{i}$ is worse than~$a$ at time $t - \epsilon$. Since $q_l^{\tau}(a_{i})$
is a linear function, we can conclude that $a_i$ never maximises
Equation~\eqref{eq:approx2} during the interval $[t-\epsilon, t]$.

\item If~$x$, which is the point at which the functions $q_l^{x}(a)$ and
$q_l^{x}(a_i)$ intersect, is greater than~$\tau$, then we add $(a_i, x)$ to~$O$.
This is because the fact that~$q_l^{x}(a_i)$ and $q_l^{x}(a)$ are linear
functions implies that~$a_i$ cannot be optimal for every time $\tau' < \tau$.

\item Finally, if~$x$ is smaller than~$\tau$, then we remove $(a, \tau)$ from
$O$ and continue by comparing~$a_i$ to the new final pair in~$O$. From the
inductive hypothesis, we have that $a$ is not optimal for every time point
$\tau' \le \tau$, and the fact that $x < \tau$ and the fact that~$q_l^{x}(a_i)$
and $q_l^{x}(a)$ are linear functions implies that $a_i$ is better than $a$ for
every time point $\tau' > \tau$. Therefore, $a$ can never be optimal.
\end{itemize}
These three observations are sufficient to prove that
Algorithm~\ref{alg:quaddamage} correctly computes $O$, and $O$ can obviously be
used to compute the approximation~$p_2$. The following lemma gives the time
complexity of the algorithm.

\begin{lemma}
Algorithm~\ref{alg:quaddamage} runs in time $O(|\Sigma| \log |\Sigma|)$.
\end{lemma}
\begin{proof}
Since sorting can be done in $O(|\Sigma| \log |\Sigma|)$ time, the first step of
this algorithm also takes $O(|\Sigma| \log |\Sigma|)$. We claim that the
remaining steps of the algorithm take $O(|\Sigma|)$ time. To see this, note that
after computing a crossing point $x$, the algorithm either adds an action to the
list~$O$, or removes an action from~$O$. Moreover each action~$a$ can enter the
list~$O$ at most once, and leave the list~$O$ at most once. Therefore at most $2
\cdot |\Sigma|$ crossing points are computed in total.
\end{proof}

We now complete the proof of Lemma~\ref{lem:sortdouble}. In order to compute the
approximation $p_2$, we simply run Algorithm~\ref{alg:quaddamage} for each location
$l \in L$, which takes $O(|L| \cdot |\Sigma| \log |\Sigma|)$ time. Finally, we
must account for the time taken to compute the approximation~$p_1$, which takes
$O(|\mathcal{M}|)$ time, as argued in Theorem~\ref{thm:singlenets}. Therefore,
we can compute~$p_2$ in time $O(|\mathcal{M}| + |L| \cdot |\Sigma| \log
|\Sigma|)$.

\section{Proof of Theorem~\ref{thm:doublenets}}

\begin{proof}
Lemma~\ref{lem:double} gives the step error for double $\epsilon$-nets to be
$\frac{1}{3}\epsilon^{3}$. Since there are $\frac{T}{\epsilon}$ intervals,
Lemma~\ref{lem:individual} implies that the global error of double
$\epsilon$-nets is $\frac{1}{3}\epsilon^{3} \cdot \frac{T}{\epsilon} =
\frac{1}{3}\epsilon^{2} \cdot T$. In order to achieve a precision of $\pi$, we
must select an $\epsilon$ that satisfies $\frac{1}{3}\epsilon^{2} \cdot T =
\pi$. Therefore, we choose $\epsilon = \sqrt{\frac{3\pi}{T}}$, which gives $T
\cdot \sqrt{\frac{T}{3\pi}}$ intervals.

The cost of computing each interval is given by Lemma~\ref{lem:sortdouble} as
$O(|\M| + |\locations|\cdot|\Sigma|\cdot \log |\Sigma|)$, and there are $T \cdot
\sqrt{\frac{T}{3\pi}}$ intervals overall, which gives the claimed complexity of
$O(|\M| \cdot T \cdot \sqrt{\frac{T}{\pi}} + |\locations| \cdot  T \cdot
\sqrt{\frac{T}{\pi}} \cdot |\Sigma| \log |\Sigma|)$.
\end{proof}

\section{Proof of Lemma~\ref{lem:knets}}

%In order to prove this lemma we will apply the arguments given for
%double-$\epsilon$ nets inductively. The proofs given for double-$\epsilon$ nets
%provide the base cases, and we must prove the inductive step.

Our arguments here are generalisations of those given for the claims made in
Section~\ref{sec:double_nets}. 

\subsection{Error bounds for the approximation $p_k$}

The following lemma is a generalisation of Lemma~\ref{lem:double}.

\begin{lemma}
\label{lem:lem12p1}
%We have $\error(k,\varepsilon) \leq
%\frac{2\varepsilon}{k+1}\error(k-1,\varepsilon)$.
For every $k > 1$, if we have $\error(k, \varepsilon) \leq c \cdot
\epsilon^{k+1}$, then we have $\error(k + 1, \varepsilon) \leq \frac{2}{k+2}
\cdot c \cdot \epsilon^{k+2}$.
\end{lemma}
\begin{proof}
The inductive hypothesis implies that $\| p_k(t -
\tau)-f_{p_{k+1}(t)}^t(t - \tau) \| \le c \cdot \tau^{k+1}$ for every $\tau \in [0,
\epsilon]$. Therefore, for every pair of locations $l,l' \in L$ and every $\tau
\in [t - \epsilon, t]$ we have:
\begin{equation*}
\|(p_k(l',t - \tau) - p_k(l, t - \tau)) - (f_{p_{k+1}(t)}^t(l', t -
\tau)-f_{p_{k+1}(t)}^t(l, t - \tau))\| \le 2 \cdot c \cdot \tau^{k+1}.
\end{equation*}
Since we are dealing with normed Markov games, we have $\sum_{l' \in
\locations}\ratematrix(l, a, l') = 1$ for every location $l \in L$ and every
action $a \in A(l)$. Therefore, we also have for every action $a$:
\begin{multline*}
\| \sum_{l'\in\locations} \ratematrix(l,a,l') \cdot \left(p_k(l',t - \tau) -
p_k(l,t - \tau)\right) \\ - \sum_{l'\in\locations} \ratematrix(l,a,l') \cdot
(f_{p_{k+1}(t)}^t(l', t - \tau)-f_{p_{k+1}(t)}^t(l, t - \tau) \| \le  2 \cdot c \cdot
\tau^{k+1}.
\end{multline*}
This implies that $\|\dot{p}_k(l,t - \tau) - \dot{f}^t_{p_{k+1}(t)}(l, t - \tau)\|
\le 2 \cdot c \tau^{k+1}$.

We can obtain the claimed result by integrating over this difference:
\begin{equation*}
\error(k + 1, \tau) = \int_0^\tau
\|\dot{p}_k(l,t - \tau) - \dot{f}^t_{p_{k+1}(t)}(l, t - \tau)\| \le
\frac{2}{k+2} \cdot c \cdot \tau^{k+2}.
\end{equation*}
Therefore, the total amount of error incurred by~$p_{k+1}$ in $[t - \epsilon,
t]$ is at most $\frac{2}{k+2} \cdot c \cdot \epsilon^{k+2}$.
%Suppose that this claim holds up to $k - 1$, which means that:
%\begin{equation*}
%\error(k - 1,\tau) := \|f_{p_{k-1}(t)}^t(t-\tau)-p_{k-1}(t-\tau)\| \le
%\frac{2 \tau}{k} \error(k-2, \tau)
%\end{equation*}
%From this we can conclude that, for every $\tau \in [0, \epsilon]$ we have:
%\begin{equation*}
%\|(p_{k-1}(l',t - \tau) - p_{k-1}(l, t - \tau)) - (f_{p_{k}(t)}^t(l', t -
%\tau)-f_{p_{k}(t)}^t(l, t - \tau))\| \le 2 \cdot \frac{2 \tau}{k}
%\error(k-2, \tau)
%\end{equation*}
%Using the arguments given in the proof for Lemma~\ref{lem:double} we can use
%this to conclude that $\|\dot{p}_k(l,t - \tau) - \dot{f}^t_{p_k(t)}(l, t -
%\tau)\| \le \tau^2$ for every $\tau \in [0, \epsilon]$. Therefore, we have:
%\begin{align*}
%\|\dot{p}_k(l,t - \tau) - \dot{f}^t_{p_k(t)}(l, t - \tau)\| &= \int_0^\tau
%\|\dot{p}_k(l,t - \tau) - \dot{f}^t_{p_k(t)}(l, t - \tau)\| d\tau \\
%&\le \int_0^\tau 2 \cdot \frac{2 \tau}{k} \error(k-2, \tau) d\tau. 
%\end{align*}
%Since $\error(k-2, \tau) = c \cdot \tau^{k-1}$ for some constant~$c$, we
%have that:
%\begin{equation*}
%\int_0^\tau 2 \cdot \frac{2 \tau}{k} \error(k-2, \tau) d\tau = \frac{2
%\tau}{k+1} \frac{2 \tau}{k} \error(k-2, \tau) = \frac{2 \tau}{k+1} \cdot
%\error(k-1, \tau).
%\end{equation*}
%Therefore, we have $\error(k,\varepsilon) \leq
%\frac{2\varepsilon}{k+1}\error(k-1,\varepsilon)$.
\end{proof}

\subsection{Error bounds for the approximation $g_2$}

We will prove the claim for the reachability player, because the proof for the
safety player is entirely symmetric. We begin by defining the approximation
$g_2$, which gives the time-bounded reachability probability when the
reachability player follows the actions chosen by~$p_k$. If $a_l^\tau$ is the
action that maximises Equation~\eqref{eq:approxk} at the location $l$ for the
time point $\tau \in [t - \epsilon, t]$  then we define $g_k(l, \tau)$ as:
\begin{align}
-\dot{g_k}(l,\tau) &= \sum_{l'\in\locations} \mathbf{Q}(l,a_l^{\tau},l') \cdot
g_k(l',\tau) &\text{if $l \in \locations_r$,} \\
-\dot{g_k}(l,\tau) &= 
\min_{a\in\Sigma(l)} \sum_{l'\in\locations} \mathbf{Q}(l,a,l') \cdot
g_k(l',\tau) &\text{if $l \in \locations_s$.}
\end{align}

Our approach to proving error bounds for $g_k$ follows the approach that we used
in the proof of Lemma~\ref{lem:doubleConcrete}. We will consider the following
class of strategies: play the action chosen by~$p_k$ for the first~$i$
transitions, and then play the action chosen by~$p_1$ for the remainder of the
interval. We will denote the reachability probability obtained by this strategy
as $g_k^i$, and we will denote the error of this strategy as $\error^{i}_s(k,
\epsilon) := \|g_k^i(t - \epsilon) - f^t_{p_2(t)}(t - \epsilon)\|$. Clearly,
as~$i$ approaches infinity, we have that~$g_k^i$ approaches~$g_k$, and
$\error^{i}_s(k, \epsilon)$ approaches $\error_s(k, \epsilon)$. Therefore, if a
bound can be established on $\error^{i}_s(k, \epsilon)$ for all $i$, then that
bound also holds for $\error_s(k, \epsilon)$.

We have by assumption that $\error(k, \varepsilon) \leq c \cdot \epsilon^{k+1}$
and $\error_s(k, \varepsilon) \leq d \cdot \epsilon^{k+1}$, and our goal is to
prove that $\error_s(k+1, \varepsilon) \leq \frac{8c + 3d}{k+2} \cdot
\epsilon^{k+2}$. We will prove error bounds on~$g_{k+1}^i$ by induction. The
following lemma considers the base case, where~$i = 1$. In other words, it
considers the strategy that plays the action chosen by~$p_{k+1}$ for the first
transition, and then plays the action chosen by~$p_{k}$ for the rest of the
interval.

\begin{lemma}
\label{lem:kcbase}
If $\varepsilon \le 1$, $\error(k, \varepsilon) \leq c \cdot \epsilon^{k+1}$,
and $\error_s(k, \varepsilon) \leq d \cdot \epsilon^{k+1}$, then we have
$\error^{1}_s(k+1, \epsilon) \le \frac{4c + 3d}{k+2} \cdot \epsilon^{k+2}$.
\end{lemma}
\begin{proof}
Suppose that the first discrete transition occurs at time $t - \tau$, where
$\tau \in [0, \epsilon]$. Let~$l$ be a location belonging to the reachability
player, and let $a^{t-\tau}_l$ be the action that maximises~$p_{k+1}$ at time $t
- \tau$. By definition, we know that the probability of moving to a
location~$l'$ is given by $\ratematrix(l, a_l^{t-\tau}, l')$, and we know that
the time-bounded reachability probabilities for each state $l'$ are given by
$g_k(l', t-\tau)$. Therefore, the outcome of choosing $a_l^{t-\tau}$ at time $t-
\tau$ is $\sum_{l'\in\locations} \ratematrix(l, a_l^{t-\tau}, l')g_k(l', t -
\tau)$. If $a^*$ is an action that would be chosen by $f^t_{p_{k+1}(t)}$ at time
$t-\tau$, then we have the following bounds:
\begin{align*}
&\| \sum_{l'\in\locations} \ratematrix(l, a_l^{t-\tau}, l')g_k(l', t - \tau) -
\sum_{l'\in\locations} \ratematrix(l, a^*, l') f^t_{p_{k+1}(t)}(l',t - \tau) \|
\\
\le &\| \sum_{l'\in\locations} \ratematrix(l, a_l^{t-\tau}, l')p_k(l', t - \tau) -
\sum_{l'\in\locations} \ratematrix(l, a^*, l') f^t_{p_{k+1}(t)}(l',t - \tau) \|
+ c \cdot \tau^{k+1} + d \cdot \tau^{k+1} \\
\le& 4 \cdot c \cdot \tau^{k+1} + d \cdot \tau^{k+1}
\end{align*}
The first inequality follows from the bounds given for $\error(k, \varepsilon)$
and $\error_s(k, \varepsilon)$. The second inequality follows from the bounds
given for $\error(k, \varepsilon)$ and Lemma~\ref{lem:dcuseful}.

Now suppose that~$l$ is a location belonging to the safety player. Since the
reachability player will follow $p_k$ during the interval $[t-\tau, t]$, we know
that the safety player will choose an action $a_g$ that minimises:
\begin{equation*}
\min_{a\in\Sigma(l)} \sum_{l'\in\locations} \mathbf{Q}(l,a,l') \cdot
g_k(l',t - \tau). 
\end{equation*}
If $a^*$ is the action chosen by $f$ at time $t - \tau$, then
the following inequality is a consequence of Lemma~\ref{lem:dcuseful}:
\begin{equation*}
\| \sum_{l'\in\locations} \ratematrix(l, a^g, l')g_k(l', t - \tau) -
\sum_{l'\in\locations} \ratematrix(l, a^*, l') f^t_{p_{k+1}(t)}(l',t - \tau) \|
\le 3 \cdot d \cdot \tau^{k+1} \\
\end{equation*}

So far we have proved that the total amount of error made by $g_2^1$ when the
first transition occurs at time $t-\tau$ is at most $(4c + 3d) \cdot
\tau^{k+1}$. To obtain error bounds for $g_{k+1}^1$ over the entire interval
$[t-\epsilon, t]$, we consider the probability that the first transition
actually occurs at time $t-\tau$:
\begin{equation*}
\error^{1}_s(k+1,\varepsilon) 
\leq \int_0^\varepsilon e^{\tau-\varepsilon} (4c + 3d) \cdot \tau^{k+1} \; d\tau
\leq  \int_0^\varepsilon (4c + 3d) \cdot \tau^{k+1} d\tau
= \frac{4c + 3d}{k+2}\varepsilon^{k+2}.
\end{equation*}
This completes the proof.
\end{proof}

%We now prove the inductive step, by considering~$g_2^k$. This is the strategy
%that follows the action chosen by~$p_2$ for the first~$k$ transitions, and then
%follows $p_1$ for the rest of the interval. This proof uses weaker bounds than
%those used in Lemma~\ref{lem:dcbase}, for reasons that will become clear during
%the proof. Note, however, that $\frac{4}{3} \cdot \tau^3 \le 2 \cdot \tau^3$,
%and therefore Lemma~\ref{lem:dcbase} can serve as a base case.

\begin{lemma}
If $\error^{i}_s(k+1, \epsilon) \le \frac{8c+3d}{k+2} \cdot \epsilon^{k+2}$ for
some~$k$ and $\error(k, \varepsilon) \leq c \cdot \epsilon^{k+1}$, then
$\error^{i+1}_s(k+1, \epsilon) \le \frac{8c + 3d}{k+2} \cdot \epsilon^{k+2}$.
\end{lemma}
\begin{proof}
The structure of this proof is similar to the proof of Lemma~\ref{lem:kcbase},
however, we must account for the fact that~$g_{k+1}^{i+1}$ follows~$g_{k+1}^i$
after the first transition rather than~$g_k$.

Suppose that we play the strategy for $g_{k+1}^{i+1}$, and that the first
discrete transition occurs at time $t - \tau$, where $\tau \in [0, \epsilon]$.
Let~$l$ be a location belonging to the reachability player, and let
$a^{t-\tau}_l$ be the action that maximises~$p_k$ at time $t - \tau$. If $a^*$
is an action that would be chosen by $f^t_{p_2(t)}$ at time $t-\tau$, then we
have the following bounds:
\begin{align*}
&\| \sum_{l'\in\locations} \ratematrix(l, a_l^{t-\tau}, l')g_{k+1}^{i}(l', t - \tau)
- \sum_{l'\in\locations} \ratematrix(l, a^*, l') f^t_{p_{k+1}(t)}(l',t - \tau)
\| \\
\le &\| \sum_{l'\in\locations} \ratematrix(l, a_l^{t-\tau}, l')p_k(l', t - \tau)
- \sum_{l'\in\locations} \ratematrix(l, a^*, l') f^t_{p_{k+1}(t)}(l',t - \tau)
\| + c \cdot \tau^{k+1} + (4c + 3d) \cdot \tau^{k+2} \\
\le & 4 c \cdot \tau^{k+1} + \frac{8c + 3d}{k+1} \cdot \tau^{k+2} \\
\le & (8c + 3d) \cdot \tau^{k+1}
\end{align*}
The first inequality follows from the inductive hypothesis, which gives bounds
on how far $g_{k+1}^{i}$ is from $f^t_{p_{k+1}(t)}$, and from the assumption
about $\error(k, \varepsilon)$. The second inequality follows from our
assumption on $\error(k, \varepsilon)$ and Lemma~\ref{lem:dcuseful}, and the
final inequality follows from the fact that $\tau \le 1$ and $k > 2$.

Now suppose that the location~$l$ belongs to the safety player. Let $a_g$ be an
action that minimises:
\begin{equation*}
\min_{a\in\Sigma(l)} \sum_{l'\in\locations} \mathbf{Q}(l,a,l') \cdot
g_{k+1}^i(l',t - \tau). 
\end{equation*}
If $a^*$ is the action chosen by~$f$ at time $t - \tau$, then
our assumption about $\error^{i}_s(k+1, \epsilon)$ and Lemma~\ref{lem:dcuseful}
imply:
\begin{equation*}
\| \sum_{l'\in\locations} \ratematrix(l, a^g, l')g_{k+1}^{i+1}(l', t - \tau) -
\sum_{l'\in\locations} \ratematrix(l, a^*, l') f^t_{p_{k+1}(t)}(l',t - \tau) \|
\le
\frac{24c + 9d}{k+2} \cdot \tau^{k+2} \le (8c + 3d) \cdot \tau^{k+1}
\end{equation*}
The first inequality follows from the inductive hypothesis and
Lemma~\ref{lem:dcuseful}, and the second inequality follows from the fact that
$\tau \le 1$ and $k > 2$.

To obtain error bounds for $g_2^{k+1}$ over the entire interval $[t-\epsilon,
t]$, we consider the probability that the first transition actually occurs at
time $t-\tau$:
\begin{equation*}
\error^{i+1}_s(k+1,\varepsilon) 
\leq \int_0^\varepsilon e^{\tau-\varepsilon} (8c + 3d) \cdot \tau^{k+1} \;d\tau 
\leq  \int_0^\varepsilon (8c + 3d) \cdot \tau^{k+1} \; d\tau
= \frac{8c + 3d}{k+2} \cdot \epsilon^{k+2}.
\end{equation*}
This completes the proof.
\end{proof}

Our two lemmas together imply that $\error^{i}_s(k+1,\varepsilon) \le \frac{8c +
3d}{k+2} \cdot \epsilon^{k+2}$ for all $i$, and hence we can conclude that
$\error_s(k+1,\varepsilon) \le \frac{8c + 3d}{k+2} \cdot \epsilon^{k+2}$. This
completes the proof of Lemma~\ref{lem:knets}.

\section{Proof of Lemma~\ref{lem:kprecision}}

\begin{proof}
Lemma~\ref{lem:knets} implies that the step error of using a $k$-net is
$\error(k,\varepsilon) \leq c \cdot \epsilon^{k+1}$ for some small constant~$c <
1$. Since we have $\frac{T}/{\epsilon}$ many intervals,
Lemma~\ref{lem:individual} implies that the global error is $T \cdot
\epsilon^k$. Therefore, to obtain a precision of $\pi$ we must choose $\epsilon
= \sqrt[k]{\frac{\pi}{T}}$.
\end{proof}

\section{Proof of Lemma~\ref{lem:piecepoly}}

\begin{proof}
We will prove this claim by induction. For the base case, we have by definition
that $p_1$ is a linear function over the interval $[t-\epsilon, t]$. For the
inductive step, assume that we have proved that $p_{k-1}$ is piecewise
polynomial with degree at most $k-1$. From this, we have that 
$\sum_{l'\in\locations} \mathbf{Q}(l,a,l') \cdot p_{k-1}$ is a piecewise
polynomial function with degree at most $k-1$ for every action~$a$, and
therefore $\opt_{a\in \Sigma(l)} \sum_{l'\in\locations} \mathbf{Q}(l,a,l')
p_{k-1}(l',\cdot)$ is also a piecewise polynomial function with degree at most
$k-1$. Since $\dot{p}_k$ is a piecewise polynomial function of degree at most
$k-1$, we have that $p_k$ is a piecewise polynomial of degree at most $k$.
\end{proof}

\section{Proof of Lemma~\ref{lem:kswitchpoint}}

\begin{proof}
Let $[t-\tau_1, t-\tau_2]$ be the boundaries of a piece in $p_{k-1}$. Since
there can be at most $\frac{1}{2}|\Sigma(l)|^{2}$ actions in the CTMG, we have
that optimum computed by Equation~\eqref{eq:approxk} must choose from at most
$\frac{1}{2}|\Sigma(l)|^2$ distinct polynomials of degree $k-1$. Since each pair
of polynomials can intersect at most~$k$ times, we have that~$p_k$ can have at
most $\frac{1}{2} \cdot k \cdot |\Sigma(l)|^2$ pieces for each location~$l$ in
the interval $[t- \tau_1, t-\tau_2]$. Since $p_{k-1}$ has~$c$ pieces in the
interval $[t-\epsilon, t]$, and $|L|$ locations, we have that~$p_k$ can have at
most $\frac{1}{2}\cdot c \cdot k \cdot |\locations| \cdot|\Sigma|^2$ during this
interval.
\end{proof}

\section{Proof of Theorem~\ref{theo:knets}}

\begin{proof}
We know that double $\epsilon$-nets can produce at most $|\Sigma|$ pieces per
interval, and therefore Lemma~\ref{lem:kswitchpoint} implies that triple
$\epsilon$-nets can produces at most $\frac{3}{2} \cdot |L| \cdot
|\Sigma|^3$ pieces per interval, and there are $T \cdot \sqrt[3]{\frac{T}{\pi}}$
many intervals. To compute each piece, we must sort $O(|\Sigma|)$ crossing
points, which takes time $O(|\Sigma|\log|\Sigma|)$. Therefore, the total amount
of time required to compute $p_3$ is $O(T \cdot \sqrt[3]{\frac{T}{\pi}} \cdot
|L| \cdot |\Sigma|^4 \cdot \log |\Sigma|)$.

For quadruple $\epsilon$-nets, Lemma~\ref{lem:kswitchpoint} implies that there
will be at most $6 \cdot |L|^2 \cdot |\Sigma|^5$ pieces per interval, and there
at most $T \cdot \sqrt[3]{\frac{T}{\pi}}$ many intervals. Therefore, we can
repeat our argument for triple $\epsilon$-nets to obtain an algorithm that runs
in time $O(T \cdot \sqrt[4]{\frac{T}{\pi}} \cdot
|L|^2 \cdot |\Sigma|^6 \cdot \log |\Sigma|)$
\end{proof}

\section{Collocation Methods for CTMDPs}

In the numerical evaluations of CTMCs, numerical methods like collocation
techniques play an important role. We briefly discuss the limits of these
methods when applied to CTMDPs, and in particular we will focus on the
Runge-Kutta method. On sufficiently smooth functions, the Runge-Kutta methods
obtain very high precision. For example, the RK4 method obtains a step error of
$O(\varepsilon^5)$ for each interval of length $\varepsilon$. However, these
results critically depend on the degree of smoothness of the functor describing
the dynamics.
To obtain this precision, the functor needs to be four times continuously
differentiable~\cite[p.157]{Hairer+Norsett+Wanner/92/SolvingODEsI}.
Unfortunately, the Bellman equations describing CTMDPs do not have this
property. In fact, the functor defined by the Bellman equations is not even once
continuously differentiable due to the $\inf$ and/or $\sup$ operators they
contain.

In this appendix we demonstrate on a simple example that the reduced precision is not merely a problem in the proof, but that the precision deteriorates once an $\inf$ or $\sup$ operator is introduced.
We then show that the effect observed in the simple example can also be observed in the Bellman equations on the example CTMDP from Figure~\ref{fig:game_example}.

Our exposition will use the notation given in
\url{http://en.wikipedia.org/wiki/Runge-Kutta\_methods} (accessed 08/04/2011).

\subsection{A Simplified Example}
Maximisation (or minimisation) in the functor that describes the dynamics of the system results in functors with limited smoothness, which breaks the proof of the precision of Runge-Kutta method (incl.\ Collocation techniques).
In order to demonstrate that this is not only a technicality in the proof of the quality of Runge-Kutta methods, we show on a simple example how the step precision deteriorates.

Using the notation of \url{http://en.wikipedia.org/wiki/Runge-Kutta\_methods}
(but dropping the dependency in $t$, that is $y'=f(y)$), consider a function
$y=(y1,y2)$ with dynamics---the functor $f$---defined by $y1'=\max\{0,y2\}$ and
$y2'=1$. Note that the functor $f$ is not partially differentiable at $(0,0)$ in
the second argument, see Figure~\ref{figure2}.

\begin{figure}%
\includegraphics[width=.49\columnwidth]{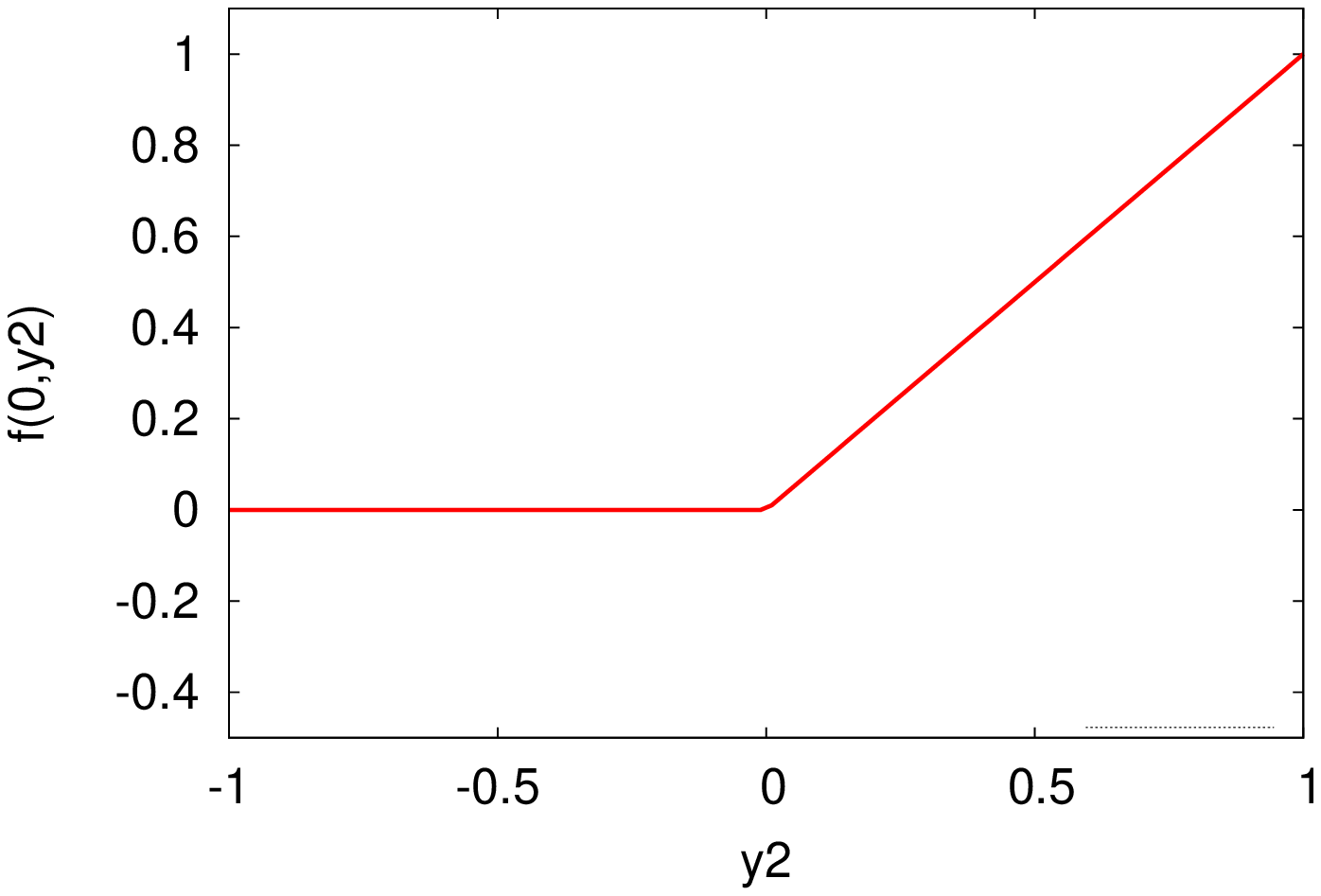}%
\includegraphics[width=.49\columnwidth]{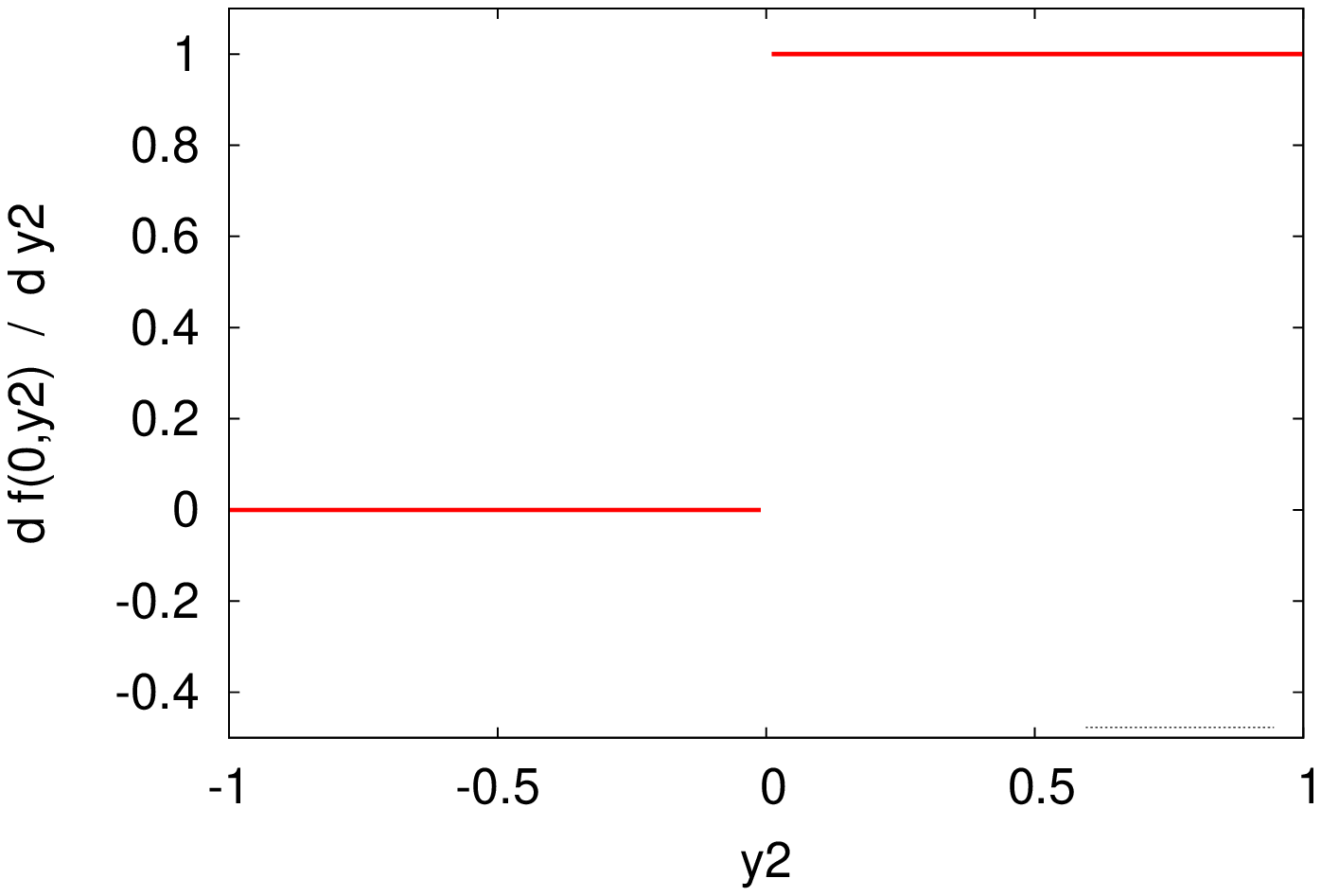}%
\caption{The left graph shows the variation of the first projection of the functor $f$ (that is, of $\max\{0,y2\}$) in the second argument (that is, of $y2$). The right graph shows the respective partial derivation in direction $y2$ on for the values on this line. In the origin (0,0) itself, $f$ is clearly not differentiable.}%
\label{figure2}%
\end{figure}

Let us study the effect this has on the Runge-Kutta method on an interval of size $h$, using the start value $y_n=(0,-\frac{1}{2}h)$.
Applying RK4, we get

\begin{figure}[b]%
\centering
\includegraphics[width=.6\columnwidth]{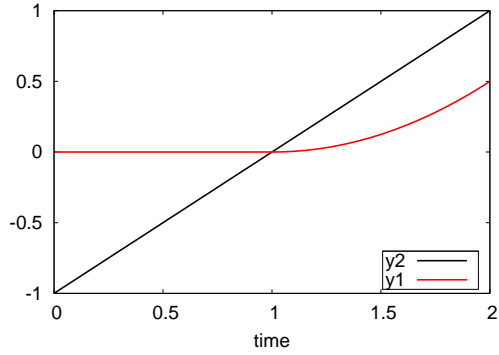}%
\caption{$y1$ and $y2$ from the solution of the ODE of the simplified example in the time interval $[0,2]$.}%
\label{figtwopaths}%
\end{figure}

\begin{itemize}
 \item $k_1=f\big((0,-\frac{1}{2}h)\big)=(0,1)$,
 \item $k_2=f\big((0,0)\big)=(0,1)$,
 \item $k_3=f\big((0,0)\big)=(0,1)$,
 \item $k_4=f\big((0,\frac{1}{2}h)\big)=(\frac{1}{2}h,1)$, and
 \smallskip
 \item $y_{n+1}=y_n + \frac{1}{6}h(k_1+2k_2+2k_3+k_4)=(h^2/12,h/2)$.
\end{itemize}

The analytical evaluation, however, provides ${\mathbf{(h^2/8,h/2)}}$  which differs from the provided result by $\mathbf{\frac{1}{24}h^2}$ in the first projection.  
Note that the expected difference in the first projection is in the order of $h^2$ if we place the point where $\max$ is in balance (the `swapping point' that is related to the point where optimal strategies change) uniformly at random at some point in the interval.

Still, one could object that we had to vary both the left and the right border of the interval. But note that, if we take the initial value
$y(0)=(0,-1)=y_0$, seek $y(2)$, and cut the interval into $2n+1$ pieces of equal length $h=\frac{2}{2n+1}$, then this is the middle interval.
(This family contains interval lengths of arbitrary small size.)

\subsection{Connection to the Bellman Equations}
The first step when applying this to the Bellman equations is to convince ourselves that their functor $F=\bigotimes_{l\in L} F_l$ with $F_l = \mathit{opt} \sum \ldots$ is indeed not differentiable. %(Provided, of course, it contains a maximisation/minimisation that can force a change in an optimal action.)
We use $g$ for the arguments of $F$ in order to distinguish it from the solution $f$, where $f(t)$ is the time-bounded reachability probability at time $t$.

For this, we simply re-use the example from Figure \ref{fig:game_example}.
The particular functor $F$ is not differentiable in the origin:
varying $F_{l_R}$ in the direction $g_l$ provides the function shown in Figure \ref{figure3}, showing that  $F_{l_R}$ is not differentiable in the origin.

(Due to the direction of the evaluation, this is the `rightmost' point where the
optimal strategy changes.)

\begin{figure}%
\includegraphics[width=.495\columnwidth]{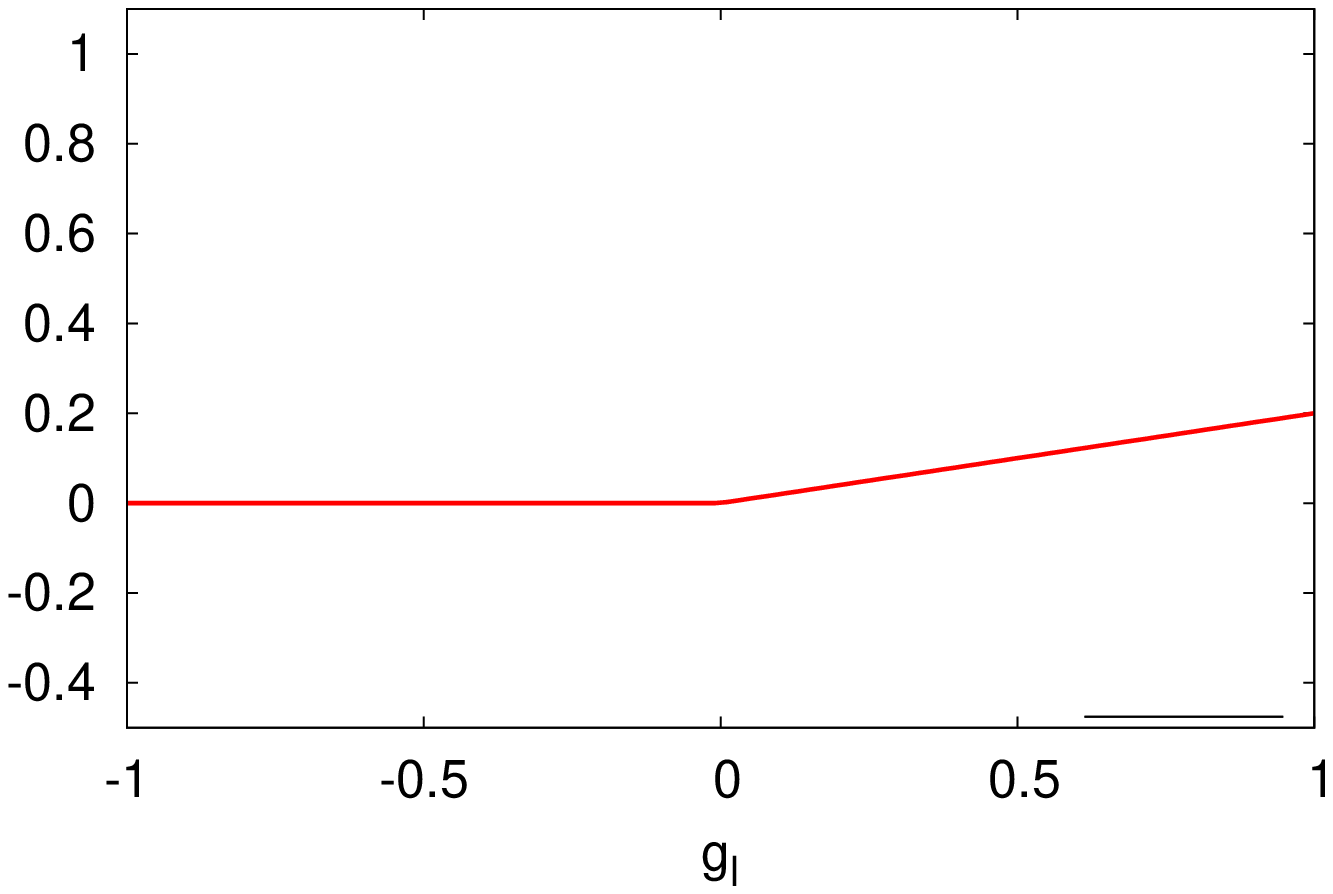}%
\includegraphics[width=.495\columnwidth]{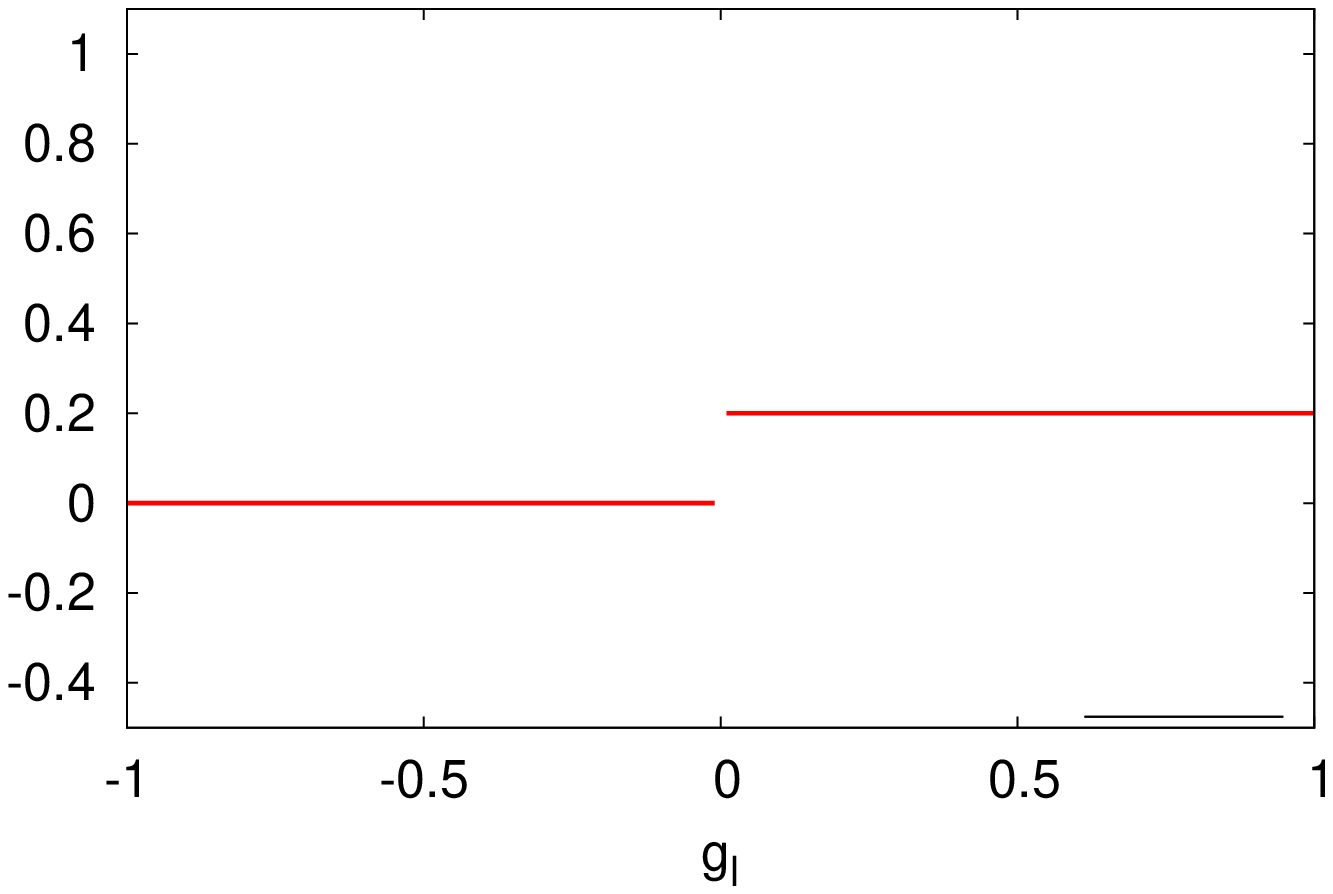}%
\caption{The left graph shows the variation of the first projection of the functor $F$ in the argument $g_l$ at the origin. The right graph shows the respective partial derivation in direction $g_l$ on for the values on this line. In the origin $\vec{0}$ itself, $F$ is clearly not differentiable.}%
\label{figure3}%
\end{figure}

Again, differentiating $F_{l_R}(f(t_1))$ in the direction $g_l$ provides a non-differentiable function. (In fact, a function similar to the function shown in Figure \ref{figure3}, but with adjusted $x$-axis.)

An analytical argument with $e$ functions is more involved than with the toy example from the previous subsection.
However, when the mesh length (or: interval size) goes towards $0$, then the 
ascent of the $e$ functions is almost constant throughout the mesh/interval.
In the limit, the effect is the same and the error in the order of $h^2$.

\end{document}